\documentclass[twocolumn,aps,pra,groupedaddress,floatfix]{revtex4-1}
\usepackage[utf8]{inputenc}

\usepackage{graphicx}  
\usepackage{bm}        
\usepackage{amssymb}   
\usepackage{amsmath,amsthm}   
\usepackage{amsfonts}  
\usepackage{physics}   
\usepackage[ruled,vlined]{algorithm2e}
\usepackage{float}
\usepackage{hyperref}
\usepackage{mathtools}
\usepackage{xcolor}

\usepackage[caption=false,subrefformat=parens,labelformat=parens]{subfig}
\usepackage{lipsum}
\usepackage{ragged2e}
\usepackage{verbatim}
\usepackage{diagbox}
\usepackage{color,xcolor,colortbl}
\usepackage{soul}
\usepackage{multirow}
\usepackage{boxedminipage}
\usepackage{changepage}
\usepackage{subfiles}

\newcommand{\leak}{\text{leak}_{\varepsilon_\EC}}
\newcommand{\prot}{Discrete-phase-randomized BB84 }

\newcommand{\QdF}{\text{QdF}}
\newcommand{\PE}{\text{PE}}
\newcommand{\PA}{\text{PA}}
\newcommand{\EC}{\text{EC}}
\newcommand{\Lin}{\mathrm{L}}

\newcommand{\Pos}{\mathrm{Pos}}

\newcommand{\Herm}{\mathrm{Herm}}

\DeclareRobustCommand{\bbone}{\text{\usefont{U}{bbold}{m}{n}1}}
\newtheorem{theorem}{Theorem}
\newtheorem{lemma}[theorem]{Lemma}

\begin{document}



\title{Numerical Calculations of Finite Key Rate for General Quantum Key Distribution Protocols }
\author{Ian George}
\author{Jie Lin}
\author{Norbert L\"utkenhaus}
\affiliation{Institute for Quantum Computing and Department of Physics and Astronomy, \\
University of Waterloo, Waterloo, Ontario, Canada N2L 3G1\\}

\date{\today}

\begin{abstract}
	 Finite key analysis of quantum key distribution (QKD) is an important tool for any QKD implementation. While much work has been done on the framework of finite key analysis, the application to individual protocols often relies on the the specific protocol being simple or highly symmetric as well as represented in small finite-dimensional Hilbert spaces. In this work, we extend our pre-existing reliable, efficient, tight, and generic numerical method for calculating the asymptotic key rate of device-dependent QKD protocols in finite-dimensional Hilbert spaces to the finite key regime using the security analysis framework of Renner. We explain how this extension preserves the reliability, efficiency, and tightness of the asymptotic method. We then explore examples which illustrate both the generality of our method as well as the importance of parameter estimation and data processing within the framework.
\end{abstract}

\pacs{}
\maketitle

\begin{section}{Introduction}

 As large scale quantum computers become an actuality, we need to change our cryptographic infrastructure to be safe against attacks which involve adversaries who have such computers at their disposal \cite{Mosca2018}. One of the cryptographic tools for this change in infrastructure is quantum key distribution (QKD), the security of which will not be threatened by future technological or algorithmic developments \cite{paterson2004,stebila2010,colbeck2011,Alleaume2014}.  See Ref. \cite{scarani09a} for a review of QKD and Refs. \cite{xu19, pirandola19} for the recent progress.  
  
A main task of the security analysis is to calculate the secret key rates that can be securely achieved with a given protocol. In analyzing QKD protocols, security proofs are often done first in the asymptotic regime, that is, in the limit of an infinite amount of quantum signals being exchanged between a sender and a receiver (traditionally known as Alice and Bob). However, in any realistic implementation of a QKD protocol, Alice and Bob can only have a finite amount of data for characterizing their channel and for performing classical post-processing. It is of practical relevance to prove composable security in the finite regime \cite{renner05} so that the key generated by QKD with properly evaluated security parameters can be used safely in other cryptographic applications such as encryption using one-time pad. Toward this goal, several protocols \cite{scarani2008,scarani08b,cai2009, curty2014, lim2014, Mizutani2015, Zhang2017, Wang2018} have been proved to be secure in the finite regime using the $\varepsilon$-security framework expounded in \cite{renner05,scarani08b}.

However, analytical methods for calculating the secret key rates are highly technical in both asymptotic and finite regimes, and they are often restricted to certain protocols with symmetry. To aid the study of more QKD protocols (especially those without symmetry) and also to study side-channel imperfections of protocols, numerical methods \cite{coles2016, winick2018, primaatmaja19, tan19} based on convex optimization and specifically semidefinite program (SDP) have been developed. In particular, numerical methods in Ref. \cite{winick2018} provide tight and reliable key rates for general finite-dimensional QKD protocols. Nevertheless, all these methods are currently restricted to the asymptotic regime. Thus, it is important to extend numerical methods to finite regime in order to preserve the advantages of numerical methods. 

In this work, we extend the numerical asymptotic key rate calculation method in Ref. \cite{winick2018} to the finite regime. For the finite key analysis, we adopt Renner's framework \cite{renner05}. Our method retains advantages of the previous numerical method \cite{winick2018}; that is, it provides a reliable lower bound on the key rate for general finite-dimensional QKD and the key rate is tight within the framework \cite{renner05}. Unlike other works \cite{cai2009,bratzik2011,bunandar2019}, our method does not make an approximation that leads to a loose bound in the parameter estimation subprotocol for certain cases. Specifically, our method remains tight when the positive-operator valued measure (POVM) used in the protocol has more than two outcomes. This makes our solver applicable for general QKD protocols. Furthermore, we show that, without changing the security parameter of the parameter estimation step, one can decrease the set of states over which one must minimize the key rate in many practical cases. We implement this improvement to the analysis of parameter estimation in our numerical method. Our numerical method also can calculate the finite key rate for protocols that accept a set of observed statistics in the parameter estimation subprotocol. This presents an opportunity that is commonly overlooked, though it is of practical relevance for actual implementations. These results differ even from a recent numerical approach to finite key analysis \cite{bunandar2019}, which was designed only for these protocols which can only achieve tight key rate for QKD protocols which use a single two-outcome POVM in parameter estimation and accept on a single observed frequency distribution, which is a restrictive case. In summary, we improve the analysis of the parameter estimation subprotocol in finite key analysis and present a reliable generic numerical method for calculating the finite key rate of QKD protocols represented in finite Hilbert spaces for the first time. 

This paper is organized as follows. In Sec. \ref{sec:background}, we review background related to finite key analysis including a review of the finite key analysis framework from Ref. \cite{renner05}. We then discuss our extension of the numerical method from Ref. \cite{winick2018} to the finite regime in Sec. \ref{sec:numMethods}. To exemplify the key ideas in our finite key analysis, we apply our method to analyze different variations of Bennett-Brassard 1984 (BB84) \cite{bennett84a} protocol including the single-photon prepare-and-measure \cite{bennett84a}, measurement-device-independent (MDI) \cite{lo2012} and discrete-phase-randomized \cite{cao2015} variants in Sec. \ref{sec:examples}. Finally we make concluding remarks in Sec. \ref{sec:conclusion}. We leave technical details in the appendices, including the derivations of the numerical method and certain improved terms in the bound on the key length.

\end{section}

\begin{section}{Background} \label{sec:background}
\subsection{General QKD Protocol in the Finite Regime}
We start by reviewing the $\varepsilon$-security framework of QKD \cite{renner05}. QKD is a cryptographic protocol for secret key distribution in which Alice and Bob establish a shared secret key by generating a pair of keys $S_{A}$ and $S_{B}$ such that the keys agree (correctness) and are completely unknown to an eavesdropper (secrecy). Neither of these properties can be achieved perfectly, so we instead talk of a QKD protocol which is  $\varepsilon = \varepsilon' + \varepsilon''$-secure as it is $\varepsilon'$-correct and $\varepsilon''$-secret where the $\varepsilon$'s quantify the amount the protocol deviates from the ideal property. A QKD protocol is $\varepsilon$-secure if a distinguisher, which is either given the real or the ideal protocol as a block box to test, can guess correctly which protocol it was given with probability at most $(1/2 + \varepsilon)$ \cite{Maurer2011,Portmann2014}. Formally, a QKD protocol is $\varepsilon$-secure if 
\begin{equation*}
    \frac{1}{2} \| \rho_{S_{A}S_{B}E} - \pi_{S_{A}S_{B}} \otimes \rho_{E} \|_{1} \leq \varepsilon
\end{equation*}
where $\pi_{S_{A}S_{B}} = \sum_{s \in \mathcal{S}} \frac{1}{|\mathcal{S}|} \ket{s}\bra{s} \otimes \ket{s}\bra{s}$, $\mathcal{S}$ is the set of secret keys the protocol could generate, and $\| \cdot \|_{1}$ is the trace norm defined as $\|A\|_{1} = \Tr(\sqrt{A^{\dag}A})$. The output secret key of a $\varepsilon$-secure QKD protocol has composable security under the abstract cryptography framework \cite{Maurer2011, Portmann2014}.

 In an entanglement-based QKD protocol, Alice (or Eve) constructs an entangled state $\rho_{AB}$. Alice and Bob then measure their respective halves of $\rho_{AB}$. In the case that Alice prepares the state, we refer to the half of the state sent from Alice to Bob as a \textit{signal}. We note that the entanglement-based description of QKD we use in this section is without loss of generality as prepare-and-measure protocols are equivalent via the source-replacement scheme \cite{curty04a,ferenczi12a} as will be reviewed in Section \ref{sec:numMethods}. 
 
 When Alice sends signals to Bob, the eavesdropper, traditionally known as Eve, has the chance to perform her attack. There are two classes of attacks generally considered in security analysis— collective and coherent. In both cases one assumes Eve has an unbounded quantum memory, so she can store all her systems indefinitely. Collective attacks assume Eve uses a new ancillary system to interact with each signal sent by Alice as it is sent across the channel after which she can measure her ancillary systems collectively whenever she should choose (even after Alice and Bob have completed their protocol). Coherent attacks assume Eve interacts with all of the signals as one large state after which she can measure whenever she pleases. As Eve interacts with all of the signals as one large state, the signals may be entangled in some arbitrary manner. As coherent attacks is the most general form of attack permitted by quantum mechanics, one ultimately needs to prove security against coherent attacks.
 
 \begin{figure}[h]
    \centering
    \includegraphics[width=0.95 \columnwidth]{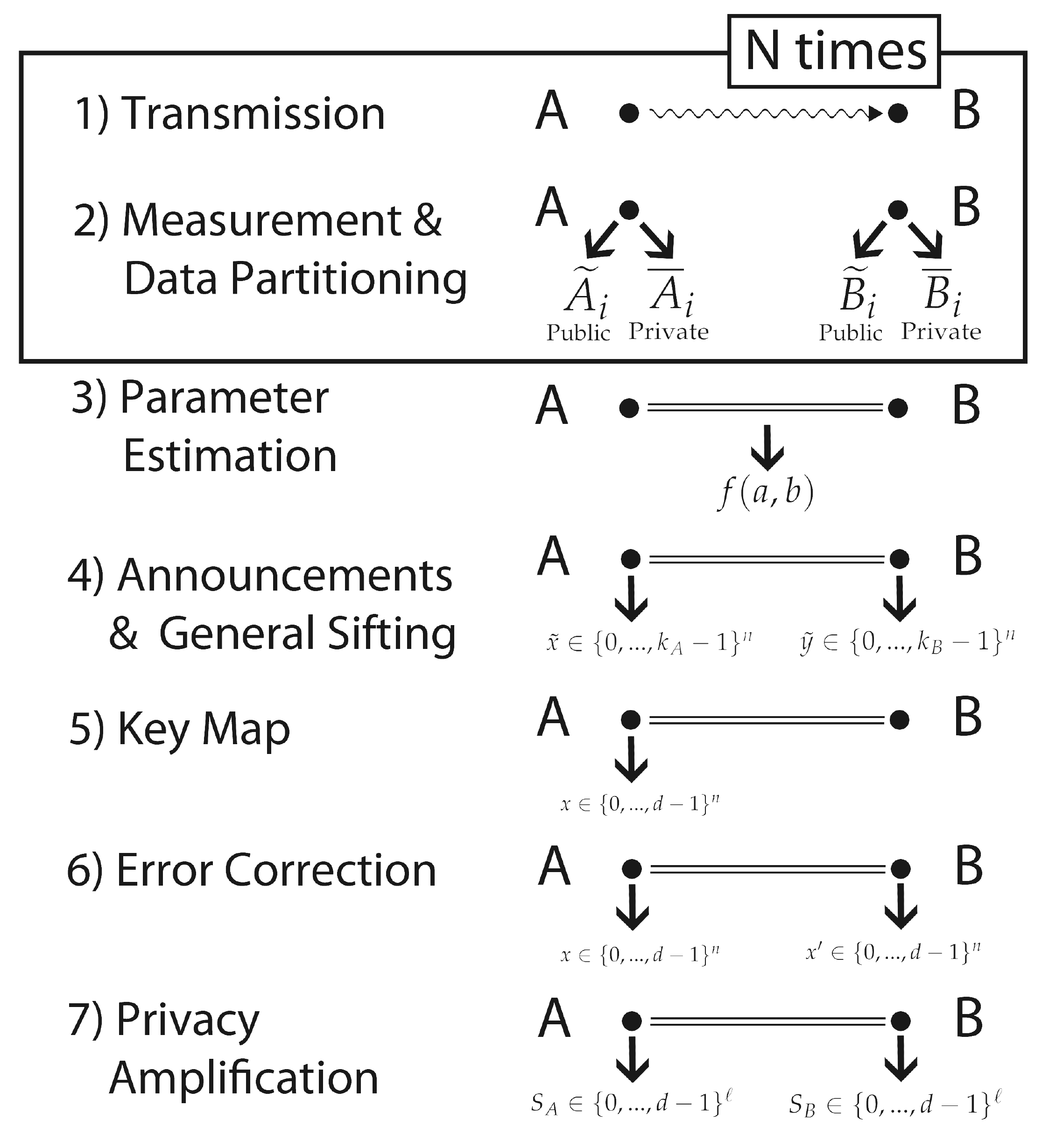}
    \caption{General QKD Protocol.}
    \label{fig1:GeneralProtocol}
\end{figure}

 With the security and attack models in mind, we can consider what subprotocols of a QKD protocol contribute to the its overall $\varepsilon$-security. To aid our discussion, we now describe steps in a general QKD protocol. Following Fig. \ref{fig1:GeneralProtocol}, without loss of generality, the general QKD protocol can be described as follows:
\begin{enumerate}
    \item \textit{State Preparation and Transmission}: Alice prepares an entangled quantum state $\rho_{AB}$ and sends half of it to Bob. Alice does this $N$ times.
    \item \textit{Measurement and Data Partitioning}: Alice and Bob measure each of the $N$ entangled quantum states $\rho_{AB}$ and store the data pertaining to each measurement. In view of future communication, they partition their respective data from each measurement, indexed by $i$, into private information, $\overline{A}_{i}, \overline{B}_{i}$, and public information $\widetilde{A}_{i}$, $\widetilde{B}_{i}$ which they later announce publicly.
    \item \textit{Parameter Estimation}: Alice and Bob announce their fine-grained data about some random subset of the $N$ signals of size $m$ to construct the frequency distribution $f(a,b)$. If $f(a,b)$ is in a set of pre-agreed upon accepted statistics, $\mathcal{Q}$, Alice and Bob proceed. Otherwise, they abort the protocol.
    \item \textit{Announcements and General Sifting}: Alice and Bob announce the public information that they prepared in Step 2 and throw out results of some subset of the $N-m$ signals based on this public information. The remaining private information forms their raw keys $\tilde{x} \in \{0,...,k_{A}-1\}^{n}$ and $\tilde{y} \in \{0,...,k_{B}-1\}^{n}$ where $k_{A}$ and $k_{B}$ are the number of possible outcomes for Alice and Bob's measurements respectively.
    \item \textit{Key Map}: Alice computes the key map,\footnote{Alternatively, Bob can compute the key map. This is commonly referred to as reverse reconciliation, and in this case Alice's and Bob's roles are reversed in steps 5. and 6.}  a function of their private data as well as the public data of both parties to obtain a key, $x \in \{0,1,...,d-1\}^{n}$ where $d$ is the size of the alphabet for the key.
    \item \textit{Error Correction}: Alice and Bob publicly communicate to try and get $\tilde{y}$ and $x$ to  agree and thus Bob obtains $x' \in \{0,1,...,d-1\}^{n}$.
    \item \textit{Privacy Amplification:} Alice and Bob produce their final keys by using a two-universal hash function on the key map result $x$ (Theorem 5.5.1 of \cite{renner05}). Privacy amplification ends with Alice and Bob having keys $S_{A}$ and $S_{B}$ respectively.
\end{enumerate}

The subprotocols which contribute to the security parameter $\varepsilon$ are parameter estimation, error correction, and privacy amplification. There is also one more source of uncertainty based on how much one `smooths' the min-entropy, $\bar{\varepsilon}$. Therefore, using the standard security proof \cite{renner05,scarani2008}, we wind up with an $\varepsilon = \varepsilon_\PE + \bar{\varepsilon} + \varepsilon_{\EC} + \varepsilon_\PA$-secure protocol which is $\varepsilon_{\EC}$-correct and $\varepsilon_{\PA}+\varepsilon_{\PE} + \bar{\varepsilon}$-secure. Each term may be viewed in the following manner:
\begin{enumerate}
    \item $\varepsilon_{\PE}$ is the probability of the parameter estimation protocol not aborting and the state which Alice and Bob tested $m$ times not being included in the security analysis.
    \item $\bar{\varepsilon}$ is the probability of Eve knowing the key because for each state feasible according to parameter estimation, $\rho_{AB}$, Alice and Bob a priori consider the min-entropy of the state $\overline{\rho}_{AB}$ that maximizes the min-entropy over the set of states $\bar{\varepsilon}$-similar to $\rho_{AB}$.
    \item $\varepsilon_{\EC}$ is the probability that Alice and Bob do not abort the protocol and obtain outputs that differ, i.e. $x \neq x'$.
    \item $\varepsilon_{\PA}$ is the probability that Alice and Bob do not abort the protocol and that the key is known to Eve because the privacy amplification failed.
\end{enumerate} 
We note each $\varepsilon$ term puts a bound on the security of Eve knowing anything about the key, which one treats as if Eve learned everything about the key. While this interpretation of the bound may seem pessimistic, depending on the data being encrypted with the key, only one bit of the original message being known may be a security threat, and so this is the appropriate security \cite{Koenig2007,benor05a}.

With the protocol described and the $\varepsilon$-terms accounted for, we can define the calculation for determining the upper bound on the length of a secret key generated by a $\varepsilon$-secure QKD protocol. We begin by defining the set of density matrices which one must minimize the key rate over given the choice that $\mathcal{Q}$ is a set of frequency distributions within some distance $t$ from a preferred, fixed frequency distribution $\overline{F}$,
\begin{align}\label{eq:paramEstSet1}
    \mathbf{S}_{\mu} = \{ \rho \in \mathcal{D}(\mathcal{H}_A \otimes \mathcal{H}_B) \hspace{0.1cm} | \exists F \in \mathcal{P}(\Sigma) \text{ such that } \\ \|\Phi_{\mathcal{P}}(\rho) - \mathcal{N}(F) \|_{1} \leq \mu \hspace{0.1cm} \& \hspace{0.1cm} \|F - \overline{F} \|_{1} \leq t\} \nonumber
\end{align}
 Throughout this work, $\mathcal{D}(\mathcal{H}_A \otimes \mathcal{H}_B)$ denotes the set of density matrices on Hilbert space $\mathcal{H}_A \otimes \mathcal{H}_B$. The map $\Phi_{\mathcal{P}}(X) \equiv \sum_{j \in \Sigma} \Tr(X \widetilde{\Gamma}_j) \dyad{j}{j}$ maps density matrices to a register of the corresponding probability distribution under the POVM, $\{\widetilde{\Gamma}_{j}\}_{j \in \Sigma}$. In other words, $\Phi_{\mathcal{P}} : D(\mathcal{H}_{A} \otimes \mathcal{H}_{B}) \to \mathcal{P}(\Sigma)$ where $\mathcal{P}(\Sigma)$ denotes the set of probability distributions over the finite set $\Sigma$ which we refer to as an alphabet. We refer to $\Phi_{\mathcal{P}}$ as the probability map. The map $\mathcal{N}(X) = \sum_{x \in \Sigma, y \in \Lambda} p(y|x) \bra{x}X\ket{x} \dyad{y}{y}$ is a map $\mathcal{N}: \mathcal{P}(\Sigma) \to \mathcal{P}(\Lambda)$ according to the conditional probability distribution $p(y|x)$. This is the quantum channel representation of a classical-to-classical channel \cite{Wilde2011}. By the data processing inequality, one knows that processing data with a map like $\mathcal{N}$ can be viewed as throwing out some information. We therefore refer to $\mathcal{N}$ as a `coarse-graining channel' within this work, as will be elaborated in the next subsection. The frequency distributions are denoted by $F,\overline{F} \in \mathcal{P}(\Sigma)$. Throughout this paper we denote POVM elements pertaining to frequency distributions we hold as being susceptible to statistical fluctuations with $\widetilde{\Gamma}$ and observables pertaining to expectations or probabilities we hold certain with $\Gamma$. Furthermore, throughout the rest of the paper, when talking about $\|P - F\|_{1}$ or $\|F-\overline{F}\|_{1}$, we will refer to these as the variational distance as $P,F$, and $\overline{F}$ are probability distributions. For this reason we refer to $\mu$ as the variation bound and $t$ as the variation threshold. 
 
 With the notation in Eqn. \ref{eq:paramEstSet1} accounted for, we see that $\mathcal{Q}$ is represented by a set of frequency distributions that have variational distance from a preferred frequency distribution $\overline{F}$ within the variation threshold, $t$, in Eqn. \ref{eq:paramEstSet1}. The other inequality in Eqn. \ref{eq:paramEstSet1} determines a limit to the variational distance between the probability distribution induced by $\rho \in D(\mathcal{H}_{A} \otimes \mathcal{H}_{B})$ under the probability map $\Phi_{\mathcal{P}}$ and the coarse-graining of some $F \in \mathcal{Q}$. In other words, Eqn. \ref{eq:paramEstSet1} determines the set of $\sigma$ which, under the map $\Phi_{\mathcal{P}}$, there exists an $F \in \mathcal{Q}$ such that the distance between $\Phi_{\mathcal{P}}$ and $\mathcal{N}(F)$ is less than the variation bound. An in-depth explanation of why this set is what one optimizes over is given in Section \ref{subsec:BackgroundPE}, but the idea is that this includes all states which lead to observations which Alice and Bob would accept with non-negligibile probability.

We note that the $\|F - \overline{F}\|_{1} \leq t$ constraint in Eqn. \ref{eq:paramEstSet1} is a choice in formalizing the set of accepted distributions during parameter estimation, $\mathcal{Q}$. While fundamentally $\mathcal{Q}$ may be any set, this threshold from some specific statistics $\overline{F}$ is a practical choice without much loss in generality as normally one would accept any probability distribution within some distance from an ideal probability distribution (such as the perfectly correlated statistics, or low phase error).

 We can now present the key rate under the assumption of identically and independently distributed (i.i.d.) collective attack. We will explain how to lift the collective attack analysis to coherent attacks in Section \ref{subsec:CoherentAttacks}.
 
 \textit{Adaptation of Theorem 6.5.1 of \cite{renner05} for Collective Attacks:}  Assuming i.i.d. collective attack, the QKD protocol is $\varepsilon = \varepsilon_\PE + \Bar{\varepsilon} + \varepsilon_\EC + \varepsilon_\PA$-secure given that, when the protocol does not abort, the output key is of length $\ell$ where
\begin{equation}\label{eq:keyLength}
    \ell \leq n( H_{\mu}(X|E) - \delta(\bar{\varepsilon})) - \text{leak}_{\varepsilon_{\EC}} - 2\log_{2}(2/\varepsilon_{\PA}) 
\end{equation}
with the following definitions:
\begin{align}
    H_{\mu}(X|E) &= \underset{\rho \in \mathbf{S}_{\mu}}{\min} H(X|E)_{\rho} \nonumber \\
    \text{leak}_{\varepsilon_{\EC}} &= n f_{\EC} H(X|Y) + \log_{2}\left(\frac{2}{\varepsilon_{\EC}}\right)  \nonumber \\
    \delta(\bar{\varepsilon}) & = 2 \log_{2}(d + 3)\sqrt{\frac{\log_{2}(2/\bar{\varepsilon})}{n}} \label{eq:delta} \\
    \mu &= \sqrt{2}\sqrt{\frac{\ln(1/\varepsilon_\PE) + |\Sigma|\ln(m + 1)}{m}} \label{eq:mu}
\end{align} 
and $d$ is the size of the alphabet for Alice and Bob's output key. 

We note that the variation bound $\mu$ is different from existing literature as we are not using an entry-wise approximation, but rather are bounding the entire variational distance and the previous statements of bounding the variational distance in general had a typo. Our $\delta(\bar{\varepsilon})$ term is smaller than any other reported work that we know of as we use the tightest bound in \cite{renner05} and using the correction noted in Footnote 27 of \cite{scarani2008}.  We note that $\leak$ as defined is an upper bound on the amount of information leaked during the error-correction step taking into account the inefficiency in the error correction for realistic block lengths using the parameter $f_{\EC} \geq 1$. In an actual QKD experiment, the information leaked is an experimentally known parameter. We derive all terms which differ from other works in Appendix \ref{appendix:Terms}.

\subsection{Parameter Estimation}\label{subsec:BackgroundPE}

It is important to consider parameter estimation's role in the security proof in greater detail as it is deceivingly simple and is the primary focus of this work's examples. In this section we clarify its role, review how it has been used in previous works, and present a theorem which resolves a standing conceptual issue.

As stated in the previous section, in parameter estimation as presented in the Renner framework \cite{renner05}, Alice and Bob sacrifice $m$ of the signals to get a sequence, $\mathbf{z} = (z_{1},z_{2}, ..., z_{m}) \in \Sigma^{m}$. From this sequence Alice and Bob construct their frequency distribution $F$ over $\Sigma$. If $F$ is in a pre-agreed set of distributions, $\mathcal{Q}$, Alice and Bob continue the protocol. Otherwise, they abort.

The $\varepsilon_{\PE}$ term in the security statement arises from disregarding any state that would lead to an accepted frequency distribution with a probability less than $\varepsilon_{\PE}$. Formally, one could say a state $\sigma$ is $\varepsilon_{\PE}$-filtered for a given set of measurements by Alice and Bob, $\{\widetilde{\Gamma}_{j}\}$, and set of accepted probabilities, $\mathcal{Q}$, if $\Pr[A_{\mathcal{T}}|\sigma] \leq \varepsilon_{\PE}$. Here $\Pr[A_{\mathcal{T}}|\sigma]$ is the probability that Alice and Bob accept a frequency distribution which is produced by sampling from $\sigma$ with the POVM defined in the protocol. A state which is ignored for this reason is referred to as being $\varepsilon_{\PE}$-filtered. This disregarding is necessary as otherwise Alice and Bob would always have to consider the maximally mixed state and be unable to generate a key.

One may note that the security statement in parameter estimation is therefore about all statistics which Alice and Bob would accept as can be formally seen in Eqn. \ref{eq:paramEstSet1}. This has been obfuscated in many of the works on finite key analysis where the security is always implicitly presented for a protocol in which only one frequency distribution is accepted. We refer to such a protocol as a \textit{protocol with unique acceptance} as there is a unique frequency distribution which Alice and Bob will accept.  While rigorous, we believe the security analysis of protocols with unique acceptance to not be the complete picture as a protocol which only accepts a single frequency distribution will abort an impractical amount of the time.

\begin{figure}
    \centering
    \includegraphics[width=\columnwidth]{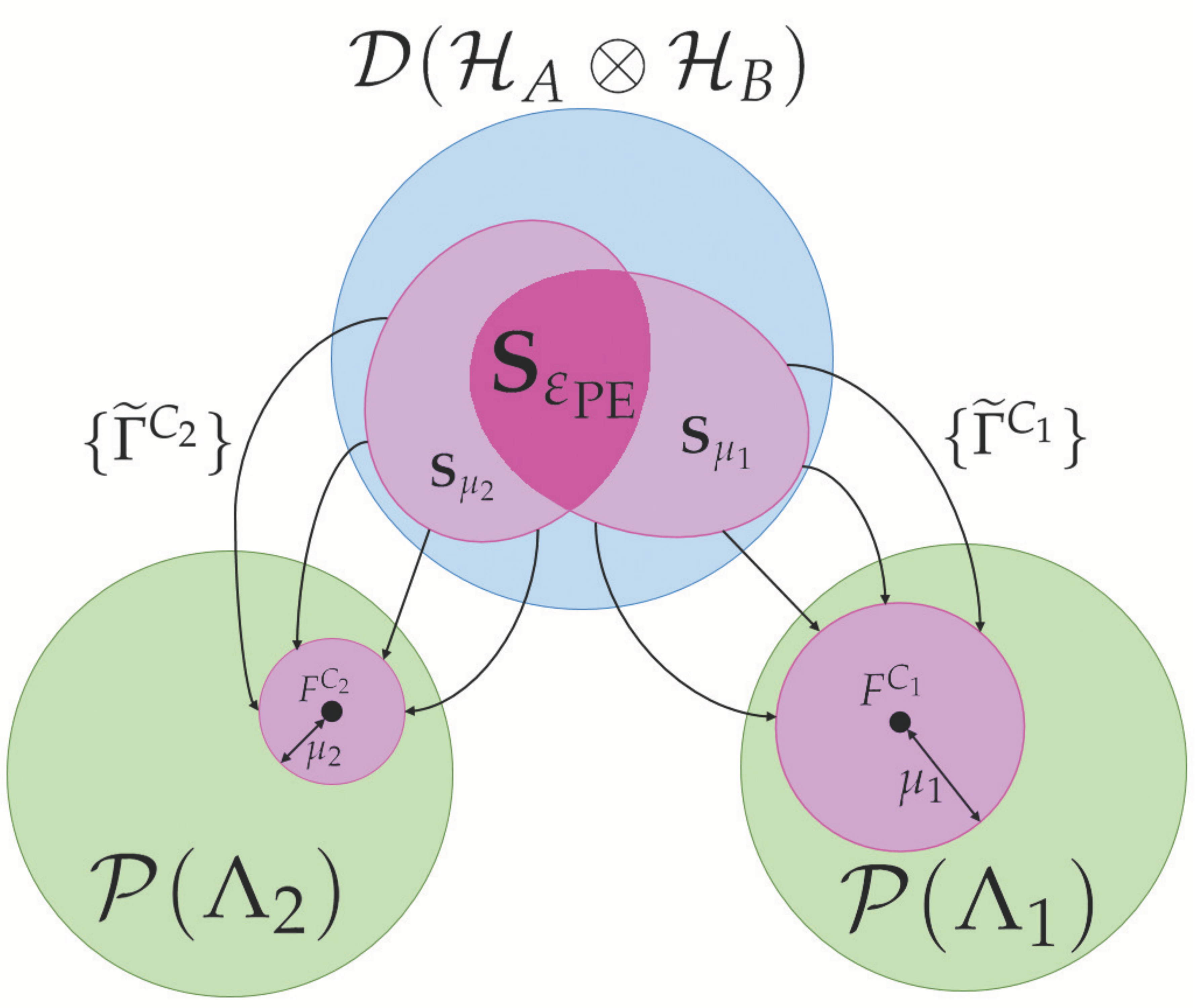}
    \caption{Parameter Estimation Visualized: Here we consider parameter estimation for two coarse-grainings. For clarity we consider a protocol with unique acceptance which only accepts $F$. Each coarse-graining of $F$ has its own set $\mathbf{S}_{\mu_{k}}$ which would correspond to Eqn. \ref{eq:paramEstSet1} for that given coarse-graining. One can see that the POVM and the variational bound $\mu_{k}$ determines the set $\mathbf{S}_{\mu_{k}}$. Furthermore, we  now visually see how by considering both coarse-grainings, $\mathbf{S}_{\varepsilon_{\PE}}$, one can decrease the set of density matrices optimized over for the same security parameter and therefore possibly increase the keyrate.}
    \label{fig:PEPic}
\end{figure}

\subsubsection{Coarse-Graining}
There remains a further conceptual issue in parameter estimation. In parameter estimation's most straightforward implementation, Alice and Bob simply take the outcomes of their joint measurements for some subset of the $N$ signals to get their sequence $\mathbf{z}$ and thus their probability distribution $F$. We refer to the sequence $\mathbf{z}$ as \textit{fine-grained} data as it pertains to the most detailed information one can acquire via the measurements permitted by the protocol. However, Alice and Bob could also construct a variety of alternative distributions by \textit{coarse-graining} the fine-grained probability distribution  $F$ over the alphabet $\Sigma$ to a probability distribution $F^{C}$ over a smaller alphabet $\Lambda$. Formally, coarse-graining is simply data processing of the statistics $F$ using a conditional probability distribution $p_{\Lambda|\Sigma}$ which is represented in the language of quantum channels as the classical-to-classical channel $\mathcal{N}$. Therefore one can construct the coarse-grained statistics $F^{C}$ and corresponding \textit{effective} POVM $\{\widetilde{\Gamma}^{C}_{i}\}_{i \in \Lambda}$ for constructing the probability map using the conditional probability distribution $p_{\Lambda|\Sigma}$ by the following two equations:
\begin{align*}
            F^{C} &= \mathcal{N}(F) & \widetilde{\Gamma}^{C}_{i} &= \sum_{j \in \Sigma} p_{\Lambda|\Sigma}(i|j)\widetilde{\Gamma}_{j}
\end{align*}

As an example, consider the case of BB84. Alice and Bob both have four possible outcomes for their measurements ``$0$", ``$1$", ``$+$", and ``$-$", which results in sixteen possible joint outcomes, which would be our alphabet $\Sigma$. However, it is often sufficient to look at a statistic known as the phase error for determining the calculation of the entropy term $H_{\mu}(X|E)$. There is a phase error if Alice and Bob's joint outcome is in the set $\Sigma_{\text{err}} := \{(``+",``-"), (``-",``+") \}$. Then the phase error can be seen as the coarse-graining from applying the conditional probability distribution $p_{\Lambda|\Sigma}$ defined as:
\begin{align*}
p_{\Lambda|\Sigma}(``\text{error}"|j) &= 
\begin{cases} 1 & j \in \Sigma_{\text{err}} \\
0 & \text{ otherwise}
\end{cases} \\
p_{\Lambda|\Sigma}(``\text{no error}"|j) &=
\begin{cases} 0 & j \in \Sigma_{\text{err}} \\
1 & \text{ otherwise}
\end{cases} \ .
\end{align*}

\subsubsection{Security with Multiple Coarse-Grainings}
Given the proof method for constructing the set in Eqn. \ref{eq:paramEstSet1}, coarse-graining may lead to a better a key rate than using just the fine-grained data as will be shown in Section \ref{sec:examples}. This would imply that, within the proof method, throwing out information can make one more secure against Eve which is counter-intuitive. However, in an actual protocol, even when Alice and Bob coarse-grain their statistics, they still have access to the fine-grained data. We would expect therefore that one can construct a set which considers the fine-grained data and the coarse-grained data and maintains the same security statement.\footnote{We note that in \cite{cai2009} they considered fine-grained data and coarse-grained data by increasing the security term.} Such a set could only improve the key rate and would resolve the idea that throwing out information can help within this proof method. Here we prove such a set exists by taking the intersection of sets constructed via different coarse-grainings but with the same security promise under the assumption of i.i.d. collective attacks. That is to say, we prove that if one fixes a parameter estimation security parameter $\varepsilon_{\PE} > 0$ and consider a finite number of coarse-grainings, indexed with an alphabet $\Xi$, then if one defines the set of states $\mathbf{S}_{\mu_{k}}$ which must be optimized over for each coarse-graining given the security parameter, then optimizing over the intersection of these sets will guarantee the same security parameter. A generalization of this theorem for considering the intersection of any set of sets, $\{\mathbf{S}_{k}\}$, such that $\forall k$ the set's complement, $\overline{\mathbf{S}}_{k}$, includes only states $\xi$ such that $\Pr[A_{\mathcal{T}}|\xi] \leq \varepsilon_{\PE}$ is straightforward.

\begin{theorem}[Security with Multiple Coarse-Grainings]\label{thm:multCoarseGrain}
Fix $\varepsilon_{\PE} > 0$. Let $\Xi$ be a finite alphabet indexing these multiple coarse-grainings. For each $k \in \Xi$, let
\begin{align*}
    \mathbf{S}_{\mu_{k}} = \{\rho \in D(\mathcal{H}_{A} \otimes \mathcal{H}_{B}) \hspace{0.1cm} |& \hspace{0.1cm} \exists F_{k} \in \mathcal{P}(\Sigma) : \\ \|\Phi_{\mathcal{P}_{k}}(\rho) - \mathcal{N}_{k}(F_{k}) \|_{1} \leq \mu_{k} 
    & \hspace{0.1cm} \& \hspace{0.1cm} \|\overline{\mathcal{N}}(F_{k}) - \overline{\mathcal{N}}(\overline{F}) \|_{1} \leq t\}    
\end{align*}
where $\overline{F} \in \mathcal{P}(\Sigma)$ is used to define the set of statistics accepted, $\Phi_{\mathcal{P}_{k}}(\rho) = \sum_{i \in \Lambda_{k}} \Tr(\rho \widetilde{\Gamma}^{C_{k}}_{i}) \dyad{i}{i}$ is the corresponding probability map, $\mathcal{N}_{k}(X) = \sum_{i,j} p_{\Lambda_{k}|\Sigma}(j|i) \bra{i} X \ket{i} \dyad{j}{j}$ is the corresponding coarse-graining channel,  $\overline{\mathcal{N}}$ is a coarse-graining channel used so that one can abort on statistics that differ from the ones considered for the variation bounds $\mu_{k}$, and $\mu_{k}$ is determined using Eqn. \ref{eq:mu} so that, by Theorem \ref{thm:muProof}, $\forall \sigma \not \in \mathbf{S}_{\mu_{k}}, \Pr[A_{\mathcal{T}}|\sigma^{\otimes m}] \leq \varepsilon_{\PE}$. Define $\mathbf{S}_{\text{multi}} = \bigcap_{k} \mathbf{S}_{\mu_{k}}$. If $\sigma \not \in \mathbf{S}_{\text{multi}}$, then $\Pr[A_{\mathcal{T}}|\sigma^{\otimes m}] \leq \varepsilon_{\PE}$.
\end{theorem}
\begin{proof}
 For any set $X \subseteq D(\mathcal{H}_{A} \otimes \mathcal{H}_{B})$, let $\overline{X} \equiv \{x \in D(\mathcal{H}_{A} \otimes \mathcal{H}_{B}) | x \not \in X \}$. We know by the construction of the sets $\mathbf{S}_{\mu_{k}}$ (Eqn. \ref{eq:mu} and Theorem \ref{thm:muProof}) that $\forall k \in \Xi$, $\forall \sigma \in \overline{\mathbf{S}_{\mu_{k}}}$, $\Pr[A_{\mathcal{T}}|\sigma^{\otimes m}] \leq \varepsilon_{\PE}$. It immediately follows that $\forall \sigma \in \bigcup_{k} \overline{\mathbf{S}_{\mu_{k}}}, \Pr[A_{\mathcal{T}}|\sigma^{\otimes m}] \leq \varepsilon_{\PE}$. Note that $\overline{\bigcap_{k} \mathbf{S}_{\mu_{k}}} = \bigcup_{k} \overline{\mathbf{S}_{\mu_k}}$. Therefore $\forall \sigma \not \in \bigcap_{k} \mathbf{S}_{\mu_{k}},$ $\Pr[A_{\mathcal{T}}|\sigma^{\otimes m}] \leq \varepsilon_\PE$.
\end{proof}

With the preceding theorem, we can define the general set to optimize over:
\begin{align}\label{eq:paramEstSetGeneral}
    \mathbf{S}_{\varepsilon_{\PE}} = \{ \rho \in D(\mathcal{H}_{A} \otimes \mathcal{H}_{B}) \hspace{0.05cm} | \hspace{0.05cm} \forall k \in \Xi, \exists F_{k} \in \mathcal{P}(\Sigma) : \\ 
    \|\Phi_{\mathcal{P}_{k}}(\rho) - \mathcal{N}_{k}(F_{k})\|_{1} \leq \mu_{k}
    \hspace{0.25cm} \& \hspace{0.05cm} \|\overline{\mathcal{N}}(F_{k}) - \overline{\mathcal{N}}(\overline{F}) \|_{1} \leq t  \} \nonumber
\end{align}
where $\Xi$ is an alphabet for indexing the number of coarse-grainings. Note that $\overline{F}$ is fixed. 

There are two important observations to be made. The first is that for $\rho \in \mathbf{S}_{\varepsilon_{\PE}}$ one does not need a single $F \in \mathcal{Q}$ which satisfies all $k$ variation bounds with respect to $\rho$ but rather $(F_{1}, .... , F_{|\Xi|}) \in \mathcal{Q}^{ \times |\Xi|}$ so that $F_{k}$ satisfies the $k^{\text{th}}$ variation bound with respect to $\rho$. This is a property of the proof method we have used as we intersect the sets. An alternative proof method that only considers one $F$ that satisfies all constraints at the same time remains an open problem. The second observation is that to define $\mathbf{S}_{\varepsilon_{\PE}}$ each $\mathbf{S}_{\mu_{k}}$ being intersected must be defined using a coarse-graining which acts on fine-grained statistics over the same alphabet $\Sigma$. Otherwise more testing would be necessary which would relate to a different set and a different security claim. To visualize Eqn. \ref{eq:paramEstSetGeneral}, see Fig. \ref{fig:PEPic}, which presents Eqn. \ref{eq:paramEstSetGeneral} for a protocol with unique acceptance.

\subsection{Asymptotic Analysis}
Lastly we review how the asymptotic analysis arises from finite key analysis since the general numerical framework for finite key analysis is an extension of the asymptotic method. This can be seen as follows. Define the asymptotic key rate as $R^{\infty} = \lim_{N \to \infty} \ell(N)/N$. As the total number of signals sent, $N$, goes to infinity, the number of signals used for parameter estimation, $m$, will grow (although an increasingly smaller \textit{fraction} of the total signals sent will be consumed for parameter estimation). Given Eqn. \ref{eq:mu}, as the the number of signals $m$ increases to infinity, the variation bound $\mu$ will go to 0. The fundamental limit for $n/N$ will be the probability that any signal can actually be used for key generation, which we refer to as $p_{pass}$. It is then clear that the asymptotic key rate is
\begin{equation}\label{eq1:DW-formula}
        R^{\infty} = p_{pass}( \underset{\rho \in \mathbf{S}}{\min} H(X|E) - f_{\EC} H(X|Y))
\end{equation}
where
\begin{equation*}
    \mathbf{S} \equiv \{ \rho \in \mathcal{D}(\mathcal{H}_{A} \otimes \mathcal{H}_{B}) | \Tr(\rho \Gamma_{i}) = \gamma_i, \hspace{0.2cm} \forall i \in \Lambda\}
\end{equation*}
and $\Lambda$ is an alphabet for indexing the constraints. \\

This statement is equivalent to the famed Devetak-Winter bound \cite{devetak05a} which in this case has been derived from the finite key analysis. Furthermore, as we will see in Section \ref{subsec:CoherentAttacks}, the finite key analysis can be extended to take into account coherent attacks and still achieve this bound in the limit. We can therefore conclude asymptotic analysis pertains to coherent attacks. In the expression for asymptotic key rates, $\{\Gamma_{i}\}$ no longer need to form a POVM, but rather can be any observables that are in the space spanned by the original POVM. This is the case as in the asymptotic limit there are no fluctuations, so one can calculate the expectation value of any observable from a linear combination of the probabilities determined by the POVM. Lastly, there exists a numerical method for calculating this key rate using semidefinite programming for general QKD protocols \cite{coles2016,winick2018}. In what follows, we show how to extend this numerical method to finite key so that we can then investigate examples to better understand parameter estimation. 

\end{section}

\begin{section}{Numerical Method}\label{sec:numMethods}

    To be able to determine the key rate for an arbitrary device-dependent QKD protocol using a unified numerical method, it is important to be able to represent all protocols in the same manner. All QKD protocols can be formulated as entanglement based protocols using the source-replacement scheme. This means that as our numerical framework can handle entanglement-based protocols, it can also handle prepare-and-measure protocols. First we review the source-replacement scheme. We then review the numerical method for asymptotic analysis under this representation. Lastly, we show how to extend the numerical analysis to consider the finite key regime.
    
    \subsection{Source-replacement Scheme}
    The source replacement scheme is a formulation of the prepare and measure protocol in the language of entanglement-based protocols. It was first made use of in the analysis of BB84 \cite{bennett92c} and Gaussian CV-QKD \cite{grosshans2003}. The general method for the equivalence was then expounded in \cite{curty04a,ferenczi12a}. By formulating the prepare and measure protocol in language of entanglement-based protocols, whatever the key rate is for the entanglement-based protocol is also the key rate for the original prepare-and-measure protocol.
    
    Imagine a prepare-and-measure protocol in which Alice sends the ensemble $\{p_{x}, \ket{\varphi_{x}}\}$ where $p_{x}$ is the a priori probability of sending the signal state $\ket{\varphi_{x}}$. By the source-replacement scheme, it is equivalent for Alice to prepare the entangled state: 
    \begin{align*}
        \ket{\Phi}_{AS} = \sum_{x} \sqrt{p_{x}} \ket{x}_{A}\ket{\varphi_{x}}_{S}
    \end{align*} 
    Alice first sends Bob's portion of the state, the signal space $S$, to Bob through a quantum channel $\mathcal{E}$ leading to the resulting joint state:
    \begin{equation*}
        \rho_{AB}=(\mathcal{I}_{A} \otimes \mathcal{E}_{S \to B})(\ket{\Phi}\bra{\Phi}_{AS})
    \end{equation*}
    where $\mathcal{I}_{A}$ is the identity channel on the $A$ space. After Alice performs a local projective measurement on the $A$ space, she effectively sends $\ket{\varphi_x}$ to Bob with probability $p_x$ just like in the prepare-and-measure scheme. Consequently Bob receives the conditional state
    \begin{equation*}
        \rho^{x}_{B} = \frac{1}{p_x}\Tr_{A}[\rho_{AB} (\ket{x}\bra{x} \otimes \bbone_{B})]
    \end{equation*}
    Assume that in the original prepare-and-measure protocol Alice and Bob ended up with a joint probability distribution $p(x,b)$ where $b \in \Sigma$ and $|\Sigma|$ is the number of POVM elements for Bob's POVM $\{\Gamma^{B}_{b}\}_{b \in \Sigma}$. It follows by the source-replacement scheme that asymptotically it is equivalent for us to constrain $\rho_{AB}$ by
    \begin{align*}
        \Tr ( \rho_{AB} ( \ket{x} \bra{x} \otimes \Gamma^{B}_{b})) = p(x,b) \hspace{0.2 cm} \forall x,b
    \end{align*}
    when minimizing $H(X|E)$ over the set $\mathbf{S}$ of compatible states.

    \subsection{Asymptotic Numerics}
   To calculate secret key rates, we have to minimize $H(X|E)$ with the given constraints on the underlying state. This is often difficult when there are not sufficient symmetries to simplify the problem. To address this issue, a two-step method to produce a tight, efficient, and reliable lower bound on $H(X|E)$ has been created \cite{winick2018}. In this work there are a few key ideas which will be of particular import in our extension to include finite size effects. The first is that $H(X|E)$ can be represented by the relative entropy as $X$ is classical information \cite{coles2012}. This is done using following function:
    \begin{equation}\label{eq:objectiveFunc}
        f(\rho) = D(\mathcal{G}(\rho)\|\mathcal{Z}(\mathcal{G}(\rho)))
    \end{equation}
     where $D(\cdot || \cdot)$ is the quantum relative entropy, $\mathcal{G}$ is a completely positive trace non-increasing map that describes the post-processing steps of the protocol and $\mathcal{Z}$ is a quantum pinching channel which is related to obtaining the results of key map (see Appendix A of \cite{lin2019} for further detail). By the joint convexity of quantum relative entropy, the function $f(\rho)$ is a convex function in $\rho$ and thus can be used as the objective function for a semidefinite program for our minimization of the conditional entropy. Therefore we define
     \begin{align}\label{eq:alphaDef}
         \alpha \equiv \underset{\rho \in \mathbf{S}}{\min} f(\rho)
     \end{align}
     
     However, as we want a lower bound that also holds if our numerical optimization routines returns before reaching the true mathematical minimum, we need to acquire the dual problem of the SDP so that we have a maximization problem. This would guarantee the computer always returns an answer approaching from below the true minimum of the conditional entropy so that we can always guarantee that our answer provides a reliable lower bound on the key rate. Unfortunately, the quantum relative entropy is a highly non-linear function and so determining the dual of this problem is difficult in general. For this reason, we linearize the function about a given density matrix. We can then acquire the dual of the \textit{linearization} of the original problem SDP, $\max_{\vec{y} \in \mathbf{S^{*}}(\sigma)} \vec{\gamma} \cdot \vec{y}$ where 
     \begin{equation}\label{eq:asymptoticDual}
     \mathbf{S^{*}}(\sigma) \equiv \{\vec{y} \in \mathbb{R}^{\abs{\Lambda}} | \sum_{i} y_{i} \Gamma_{i} \leq \nabla f(\sigma)\}
     \end{equation}
     where $\vec{\gamma}$ is just the vector of the set of expectation values $\{\gamma_{i}\}_{i \in \Lambda}$.
     
     Then the lower bound for any optimal or suboptimal attack $\sigma$ can be calculated as 
     \begin{align}\label{eq:betaDef}
         \beta(\sigma) \equiv f(\sigma) - \Tr(\sigma \nabla f(\sigma)) + \underset{\vec{y} \in \mathbf{S^{*}}}{\max} \vec{\gamma} \cdot \vec{y}
     \end{align}
     because it can be shown that for all $\rho \in \mathbf{S}, \alpha \geq \beta(\rho)$ so long as $\nabla f(\rho)$ exists (Theorem 1 of \cite{winick2018}). Here we have defined the gradient of $f$ at point $\rho$ represented in the standard basis $\{\ket{k}\}$ as:\footnote{Note that we have defined the derivative differently than in \cite{winick2018} by absorbing the occuring transposition into the definition of the gradient. This removes transpositions in many equations. Every statement is kept consistent with this definition throughout the paper.}
     \begin{align*}
         \nabla f(\rho) \equiv \sum_{j,k} d_{jk} \ket{k}\bra{j}, \hspace{0.1cm} \text{with } d_{jk} \equiv \left. \pdv{f(\sigma)}{\sigma_{jk}} \right|_{\sigma = \rho}
     \end{align*}
     and $\sigma_{jk} \equiv \bra{j}\sigma\ket{k}$. Moreover, we can write the gradient of $f(\rho)$ as:
    \begin{align}\label{eq:derivative}
        \nabla f(\rho) \equiv \mathcal{G}^{\dag}(\log_{2}\mathcal{G}(\rho)) - \mathcal{G}^{\dag}(\log_{2} \mathcal{Z}(\mathcal{G}(\rho)))
    \end{align}
    Lastly, one can guarantee $\nabla f(\rho)$ exists via perturbing the state sufficiently by mixing the output of $\mathcal{G}(\rho)$ with the maximally mixed state such that all eigenvalues are non-zero.
     
     The expression of $\beta(\sigma)$ in Eqn. \ref{eq:betaDef} gives a valid lower bound for the key rate for any $\sigma$, but the bound will be tighter the closer $\sigma$ is to the true optimum. We thus use a near-optimal evaluation of the primal problem (Eqn. \ref{eq:alphaDef}) . This is referred to as Step 1 (see Algorithm \ref{alg:Method}). For further information on the specifics of this method, we refer to Appendix \ref{appendix:NumMethods} and \cite{winick2018}.

    \begin{algorithm}[h]
    \SetAlgoLined
    \KwResult{lower bound on $\underset{\rho \in \mathbf{S}}{\min}H(X|E)$ \cite{winick2018}}
    \begin{enumerate}
        \item Let $\epsilon >0$, $\rho_0 \in \mathbf{S}$, $\text{maxIter} \in \mathbb{N}$, and $i = 0$.
        \item[\textit{Step 1}]
        \item Compute $\Delta \rho := \arg \min_{\delta \rho} \Tr[(\delta \rho) \nabla f(\rho_i)]$ subject to $\Delta \rho + \rho_{i} \in \mathbf{S}$.
        \item If $\Tr[(\Delta \rho) \nabla f(\rho_i)] < \epsilon$, then proceed to Step 2
        \item Find $\lambda \in (0,1)$ that minimizes $f(\rho_i + \lambda \Delta \rho)$
        \item Set $\rho_{i+1} = \rho_i + \lambda \Delta \rho$, $i \to i + 1$.
        \item If $i > \text{maxIter}$, proceed to Step 2
        \item[\textit{Step 2}]
        \item Let $\rho$ be the result of Step 1. Let $\zeta \geq 0$ be the maximum constraint violation of $\rho$ from the original set $\mathbf{S}$ constraints which satisfy this. 
        \item Calculate $\nabla f(\rho)$ to use for constructing $\mathbf{S^{*}}$
        \item Expand $\mathbf{S^{*}}$ such that states which violated the original constraints by $\zeta$ are included.
        \item Calculate $\beta$ using the SDP defined above Equation \ref{eq:asymptoticDual}
    \end{enumerate}
    
    \caption{Asymptotic Key Rate lower bound}
    \label{alg:Method}
    \end{algorithm}
  
\begin{subsection}{Extension to Finite Key}
\begin{subsubsection}{Tight, reliable, and efficient lower bound}
We now extend the previous numerical framework to the finite regime and show rigorously that this extension preserves the advantages of previous numerical method; that is, it provides tight, efficient and reliable key rates. For clarity, we will begin by proving the tightness in the case where there are no numerical errors and a single coarse-graining. Here tightness is defined as the property that if one acquires the optimal solution in Step 1 of the algorithm, then Step 2 will obtain the same answer. We then generalize to the case for handling issues due to numerics and multiple POVMs in Appendix \ref{appendix:NumMethods}.

The primary steps in extending our method to the finite key regime are changing the sets over which we optimize and changing how we perform Item 9 in Algorithm \ref{alg:Method}. In the case of prepare-and-measure protocols, we first modify $\mathbf{S}_{\mu}$ as defined in Equation \ref{eq:paramEstSet1} as Alice knows her portion of the state $\rho_{A}$ perfectly under the source-replacement scheme. Therefore, while the parameter estimation is handled in the original definition, it must take into account Alice's certainty on $\rho_{A}$. Thus we define a variation of Eqn. \ref{eq:paramEstSet1} for prepare-and-measure protocols: 
\begin{align}\label{eq:SmuMainText}
\mathbf{S}_{\mu}^{\text{PM}} \equiv \{& \rho \in \Pos(\mathcal{H}_{A} \otimes \mathcal{H}_{B}) \hspace{0.05cm} | \hspace{0.05cm} \exists F \in \mathcal{P}(\Sigma) : \\ & \|\Phi_{\mathcal{P}}(\rho) - \mathcal{N}(F) \|_1 \leq \mu \text{, } 
\hspace{0.25cm} \|\overline{\mathcal{N}}(F) - \overline{\mathcal{N}}(\overline{F})\|_{1} \leq t, \hspace{0.2cm} \nonumber \\ & \Tr(\rho \Gamma_{i}) = \gamma_{i}, \hspace{0.1cm} \forall i \in \Lambda\} \nonumber
\end{align}
where we use $\Lambda$ for indexing constraints which are certain. We note $\rho$ is a density matrix by setting $\Gamma_{1} = \bbone_{\mathcal{H}_A \otimes \mathcal{H}_B}$, $\gamma_{1} = 1$. Furthermore, the use of $\mathcal{N}$ and $\overline{\mathcal{N}}$ is so that one might abort on some set of coarse-grained or fine-grained data which differs from the data used in relation to the variation bound. We stress that as $\overline{F}$ and $\overline{\mathcal{N}}$ are fixed, $\overline{\mathcal{N}}(\overline{F})$ is a fixed frequency distribution. From this we can define the primal problem of the linearized SDP at the density matrix $\rho$ as:
\begin{equation}\label{eq:finiteSDPPrimalSimple}
        \begin{aligned}
                & {\text{minimize}} & & \langle \nabla f(\rho), \sigma \rangle & \\
                & \text{subject to} & & \Tr(\Gamma_{i} \sigma) = \gamma_{i} & \forall i \in \Lambda \\
                &                   & & \| \Phi_{\mathcal{P}}(\sigma) - \mathcal{N}(F) \|_{1} \leq \mu  & \\
                &                   & & \| \overline{\mathcal{N}}(F) - \overline{\mathcal{N}}(\overline{F}) \|_{1} \leq t  & \\
                &                   & & \Tr(F) = 1 \\
                &                   & & \sigma, F \succeq 0
        \end{aligned}
\end{equation}
However it is not obvious from this form that this is an SDP. The trick is then to consider how to handle the trace norm. The trace norm of a Hermitian matrix $A$ has a well known semidefinite program \cite{watrous2018}:
\begin{equation*}
        \begin{aligned}
                & {\text{minimize}} & & \Tr(Q) + \Tr(R) & \\
                & \text{subject to} & & Q \succeq A \\
                &                   & & R \succeq -A \\
                &                   & & Q,R \succeq 0
        \end{aligned}
\end{equation*}
 It is known that the trace norm SDP always achieves the optimal value in both the primal and dual, which is a property known as strong duality. This is important as we need our SDP to have strong duality for tightness (see Appendix \ref{appendix:NumMethods} for more details). With this knowledge, we can express our SDP:
\begin{equation}\label{eq:finiteSDPPrimal}
        \begin{aligned}
                & {\text{minimize}} & & \langle \nabla f(\rho), \sigma \rangle & \\
                & \text{subject to} & & \Tr(\Gamma_{i} \sigma) = \gamma_{i} & \forall i \in \Lambda \\
                &                   & & \Tr(\Delta^{+}) + \Tr(\Delta^{-}) \leq \mu & \\
                &                   & & \Delta^{+} \succeq \Phi_{\mathcal{P}}(\sigma) - \mathcal{N}(F) & \\
                &                   & & \Delta^{-} \succeq -(\Phi_{\mathcal{P}}(\sigma) - \mathcal{N}(F)) &\\
                &                   & & \Tr(\overline{\Delta}^{+}) + \Tr(\overline{\Delta}^{-}) \leq t & \\
                &                   & & \overline{\Delta}^{+} \succeq \overline{\mathcal{N}}(F) - \overline{\mathcal{N}}(\overline{F}) & \\
                &                   & & \overline{\Delta}^{-} \succeq \overline{\mathcal{N}}(\overline{F}) - \overline{\mathcal{N}}(F) &\\
                &                   & & \Tr(F) = 1 \\
                &                   & & \sigma,F, \Delta^{+}, \Delta^{-}, \overline{\Delta}^{+}, \overline{\Delta}^{-} \succeq 0
        \end{aligned}
\end{equation}
The dual of this problem is:
\begin{equation}\label{eq:finiteSDPDual}
        \begin{aligned}
                & {\text{maximize}} & & \vec{\gamma} \cdot \vec{y} + \overline{f} \cdot \vec{\overline{z}} - a\mu - \overline{a}t - b\\
                & \text{subject to} & & \sum_i y_{i}\Gamma_i + \sum_j z_j \widetilde{\Gamma}_j \leq \nabla f(\rho)\\
                &                   & & \overrightarrow{\overline{N}^{\dagger}}(\vec{\overline{z}}) - \overrightarrow{N^{\dagger}}(\vec{z}) \leq b \vec{1} \\
                &                   & & -a \vec{1}\leq \vec{z} \leq a \vec{1} \\
                &                   & & -\overline{a} \vec{1}\leq \vec{\overline{z}} \leq \overline{a} \vec{1} \\
                &                   & &  a,\overline{a} \geq 0, \vec{y} \in \mathbb{R}^{\abs{\Lambda}} \\
        \end{aligned}
\end{equation}
where $\overline{f}$ is the vector version of $\overline{\mathcal{N}}(\overline{F})$ and $\overrightarrow{\mathcal{N}^{\dagger}}$ is the action as the adjoint of the map $\mathcal{N}$, $\mathcal{N}^{\dagger}$, on the diagonal entries of a matrix. It is sufficient to consider $\overrightarrow{\mathcal{N}^{\dagger}}$ on the diagonal entries of a matrix because $\mathcal{N}^{\dagger}$ only acts on the diagonal entries of a matrix, and so it is easy to see that the $\overrightarrow{\mathcal{N}^{\dagger}}$ map applied to the vector formed by the diagonal entries of a matrix gives the equivalent action as $N^{\dagger}$ on the matrix.

One may note that the objective function of the finite key SDP is similar to the asymptotic case but with reductions associated with the finite size effects due to the variational bound $\mu$ and the threshold $t$ as the variables $a,\overline{a}$ are non-negative. However this is somewhat obfuscated when first presented in this general form. We therefore explain this in relation to the simplified SDP of a protocol with unique acceptance in the following section. We denote the set of $(a,\overline{a},b,\vec{y},\vec{z},\vec{\overline{z}})$ which satisfy the constraints as $\mathbf{S}_{\mu}^{*}(\rho)$ for a primal solution $\rho$ to mirror the asymptotic notation.

With the SDP for finite key analysis determined, it is crucial to prove that we have preserved the old properties of tightness, robustness to perturbation to make $\nabla f(\rho)$ exist, and reliability in the face of finite computational precision. As we have not changed the function $f$, all of the theorems pertaining to perturbing the channel to guarantee $\nabla f(\rho)$ exist are unchanged from asymptotic case, and we direct readers to Ref. \cite{winick2018} for those proofs. However, the proof of tightness is not identical to that in \cite{winick2018} and so we state this result here. 
\begin{theorem}[Equality of $\alpha = \beta(\rho^*)$] \label{thm:MainTextTightness}
If $\rho^{*}$ is the minimizer that achieves $\alpha$, then $\alpha = \beta(\rho^{*})$ where $$\alpha \equiv \underset{\rho \in \mathbf{S}_{\mu}}{\min}f(\rho)$$ and 
\begin{align*}\beta(\sigma) \equiv & f(\sigma) - \Tr(\sigma \nabla f(\sigma)) \\ & \hspace{0.2cm} + \underset{(a,\overline{a},b,\vec{y},\vec{z},\vec{\overline{z}}) \in \mathbf{S}_{\mu}^{*}(\sigma)}{\max} \vec{\gamma} \cdot \vec{y} + \overline{f} \cdot \vec{\overline{z}} - a\mu - \overline{a}t - b
\end{align*}
This guarantees our numerical method obtains the optimal value when the solver works ideally.
\end{theorem}
\begin{proof}
See Appendix \ref{appendix:NumMethods}.
\end{proof}
Note this is not obvious as $\alpha$ is the optimal of the primal using the original function $f(\rho)$, and $\beta$ includes the dual of the linearization of $f(\rho)$.

Lastly, we are concerned with the numerical precision of the computer which cannot perfectly represent the POVM elements or statistics and sometimes may return an answer in the first step that slightly violates some constraint. In other words the computer has not optimized over $\mathbf{S}_{\mu}$, but rather over some different set $\mathbf{\widetilde{S}}_{\mu}$. Without handling this, our solver could be unreliable, i.e. it could allow for the solution of step 1 to obtain a value greater than step 2 in some case. To guarantee this does not happen, one must expand the set for the dual $\mathbf{S}^{*}_{\mu}(\sigma)$ to $\widetilde{\mathbf{S}}^{*}_{\mu}(\sigma)$. The proper method for doing this is to find the largest constraint violation of the certainty constraints, which we denote by $\zeta$'.\footnote{In the definition of Algorithm \ref{alg:Method}, $\zeta$ was the violation of all constraints, but all constraints were certain. Due to the uncertainty constraints in the finite-case, $\zeta$ only applies to the certainty constraints and then we handle expanding the uncertainty constraints accordingly. So as to avoid confusion, we define $\zeta'$ as the parameter pertaining only to certainty constraint violations in the finite key case.} Then one must allow every certainty constraint to vary within that distance as was done in the asymptotic case: $|\Tr(\rho \Gamma_i) - \gamma_i | \leq \zeta'$. Furthermore, one expands $\mu$ to $\mu' = \max(\mu + n  \epsilon', \| \Phi_{\mathcal{P}}(\rho_{f}) - F \|_{1} + n  \epsilon')$ where $n = |\Lambda|$ and $\rho$ is the solution to the first step. We leave the proof of this statement to Appendix \ref{appendix:NumMethods}. This then guarantees to include the state considered in the first step. Therefore, we have an SDP to do finite key analysis which is tight, efficient, and reliable for general QKD protocols.
\end{subsubsection}

\begin{subsection}{SDP for Protocol with Unique Acceptance}
Many of our examples pertain to protocols with unique acceptance for clarity in relation to previous work as well as for clarity of ideas. As in the case of unique acceptance the problem simplifies, we derive the SDP for a QKD protocol with unique acceptance from the general version above. Most generally, a protocol with unique acceptance may be viewed as picking $\overline{\mathcal{N}}(\overline{F})$ to be the only distribution Alice and Bob accept on. Then the constraint pertaining to $\mathcal{Q}$ in Eqn. \ref{eq:SmuMainText} vanishes as it must be the case $\overline{\mathcal{N}}(F) = \overline{\mathcal{N}}(\overline{F})$. It follows $F$ could be allowed to vary over all $F \in \mathcal{P}(\Sigma)$ such that $\overline{\mathcal{N}}(F) = \overline{\mathcal{N}}(\overline{F})$. However, in previous works \cite{scarani2008,scarani08b,bratzik2011,bunandar2019}, this nuance is lost as only one coarse-graining is considered, and so the authors instead define the frequency distribution on the the coarse-grained outcomes by defining $F := \overline{\mathcal{N}}(\overline{F})$. For consistency in the literature, we also make this assumption in defining a protocol with unique acceptance. We denote $\overline{\mathcal{N}}(\overline{F})$ by $\overline{F}_{\overline{\mathcal{N}}}$ to make it clear it is fixed rather than a variable. Using this notation, we can define the following set:
\begin{align*}
\mathbf{S}^{\text{UA}}_{\varepsilon_{\PE}} \equiv \{& \rho \in \Pos(\mathcal{H}_{A} \otimes \mathcal{H}_{B}) \hspace{0.05cm} | \hspace{0.05cm} \\
& \hspace{0.25cm} \|\Phi_{\mathcal{P}}(\rho) -  \mathcal{N}(\overline{F}_{\overline{\mathcal{N}}}) \|_1 \leq \mu \text{, } \Tr(\rho \Gamma_{i}) = \gamma_{i}, \hspace{0.1cm} \forall i \in \Lambda\}
\end{align*}
where it must be the case that $\mathcal{N}$ coarse-grains data from the alphabet of $\overline{F}_{\overline{\mathcal{N}}}$ as otherwise it would not be well-defined.
From this definition we get the following primal problem:
\begin{equation}\label{eq:SDfiniteSDPPrimal}
        \begin{aligned}
                & {\text{minimize}} & & \langle \nabla f(\rho), \sigma \rangle & \\
                & \text{subject to} & & \Tr(\Gamma_{i} \sigma) = \gamma_{i} & \forall i \in \Lambda \\
                &                   & & \Tr(\Delta^{+}) + \Tr(\Delta^{-}) \leq \mu & \\
                &                   & & \Delta^{+} \succeq \Phi_{\mathcal{P}}(\sigma) -  \mathcal{N}(\overline{F}_{\overline{\mathcal{N}}}) & \\
                &                   & & \Delta^{-} \succeq -(\Phi_{\mathcal{P}}(\sigma) - \mathcal{N}(\overline{F}_{\overline{\mathcal{N}}})) &\\
                &                   & & \sigma, \Delta^{+}, \Delta^{-} \succeq 0
        \end{aligned}
\end{equation}
The dual of this problem is:
\begin{equation}\label{eq:SDfiniteSDPDual}
        \begin{aligned}
                & {\text{maximize}} & & \vec{\gamma} \cdot \vec{y} + \overline{f} \cdot \vec{z} - a\mu \\
                & \text{subject to} & & \sum_i y_{i}\Gamma_i + \sum_j z_j \widetilde{\Gamma}_j \leq \nabla f(\rho)\\
                &                   & & -a \vec{1}\leq \vec{z} \leq a \vec{1} \\
                &                   & &  a \geq 0, \vec{y} \in \mathbb{R}^{\abs{\Lambda}} \\
        \end{aligned}
\end{equation}
where $\overline{f}$ is the vector version of $\mathcal{N}(\overline{F}_{\overline{\mathcal{N}}})$. 

 The SDP is nearly identical to the asymptotic case as the first constraint of Eqn. \ref{eq:finiteSDPDual} is in effect identical to the single constraint of Eqn. \ref{eq:asymptoticDual}. Similarly, the objective function is nearly identical, though one can see that there is some reduction to the key rate associated with the finite size effects, represented by variational bound $\mu$, as the variable $a$ is constrained to be non-negative. The constraint on $\vec{z}$ is simply the dual problem of the trace norm simplified using the specific structure of our problem (see Appendix \ref{appendix:NumMethods} for derivation).
\end{subsection}

\begin{subsection}{Coherent Attacks}\label{subsec:CoherentAttacks}
As one important aspect of finite key analysis is the ability to analyze the key rate using coherent attacks, it is important to understand how the numerics can handle the coherent attack analysis. Extending the numerical approach in this work to coherent attacks using the Finite Quantum de Finetti Theorem \cite{renner05} can be done by changing how one defines the variation bound $\mu$ and by adding some extra parameters, as we explain in Appendix \ref{appendix:CoherentAttacks}. However, the Finite Quantum de Finetti approach provides pessimistic key rates for realistic block sizes. An alternative method to the Finite Quantum de Finetti theorem which provides better, but still pessimistic, bounds on the key rate is the post-selection technique \cite{christandl2009}. This method effectively states that given $\varepsilon$-security for convex combinations of i.i.d. states, $\sigma^{\otimes N}$, which follows from the security of i.i.d. collective attacks, then the protocol is $\varepsilon' = (N+1)^{d_{AB}^2+1} \varepsilon$-secure for coherent attacks, where $d_{AB}$ is the dimension of the Hilbert space that Alice and Bob's joint state lives in. However, to rigorously use this method, this either requires the initial protocol to be permutation invariant, or finds a way to bound the portion of the protocol after parameter estimation by a permutation invariant version which introduces more $\varepsilon$ terms (Section 3.4.3 of \cite{beaudry2015}). Another technique that handles coherent attacks is the Entropy Accumulation theorem \cite{Dupuis2016}. However, this method is not immediately applicable to our numerical method since it requires a specific property for the protocol. We leave it as future work to investigate how to combine the entropy accumulation theorem with our numerical method. As such, the currently applicable coherent attack proof methods- the Finite Quantum de Finetti method and the post-selection technique- while implementable, are pessimistic and we expect them to be improved to be significantly closer to the collective attack results we present in this work. Moreover, this uplift is independent of our work here, so we concentrate on the collective attack. In Appendix \ref{appendix:CoherentAttacks} we explain how to lift our results to coherent attacks using the Finite Quantum de Finetti method. Lastly, to the best of our knowledge, the post-selection technique has not been rigorously applied to protocols using the source-replacement scheme. This is because the the source replacement scheme only proves protocol security on states with a fixed reduced density matrix, but the post-selection technique proof requires security on arbitrary i.i.d states. A simple solution is to treat the marginal constraints as uncertain and introduce extra testing in the protocol on the marginal, but this will come at some extra cost in the small block-size regime. This issue does not arise for the Finite Quantum de Finetti security proof as one can introduce an extra $\varepsilon$-term to handle the fixed marginal (See Remark 4.3.3. of \cite{renner05}).
\end{subsection}

\end{subsection}
\end{section}

\begin{section}{Examples}\label{sec:examples}
In this section we present variations of the BB84 protocol \cite{bennett84a} to investigate the properties of finite key analysis as  well as our method. In doing so, we show our method works for any protocol which can be represented in a finite-dimensional Hilbert space. This includes single-photon protocols including single-photon MDI protocols, and any optical implementation of a protocol which admits a squashing map \cite{beaudry08a,moroder10a,Zhang2020}. Furthermore we show the power and generality of our method in being able to consider multiple coarse-grainings where each frequency distribution can be of any length. This is in contrast to previous works \cite{cai2009,bratzik2011,bunandar2019} which could only do multiple two-outcome probability distributions without adding looseness to their calculation of $H_{\mu}(X|E)$ as their bound on the variation bound $\mu$ loosened beyond two-outcome POVMs. Lastly, in addition to examples of protocols with unique acceptance, we also present an example where $\mathcal{Q}$ is not a single distribution and discuss when using our method to calculate key rates of general protocols may be needed in the practical development of QKD hardware.

In all examples in this section, we let $\varepsilon_{\PE} = \bar{\varepsilon} = \varepsilon_{\EC} = \varepsilon_{\PA} = \frac{1}{4}\times 10^{-8}$ as we found no general asymmetric choice consistently improved the key rate substantially. We note that our method will work for significantly smaller $\varepsilon$ values. The only limitation is numerical precision which will not be a problem for any realistic $\varepsilon$ term given the equations always depend on the logarithm of the $\varepsilon$ term.

For completeness, we present the post-processing maps, $\mathcal{G}$, for each protocol in Appendix \ref{app:postProc} which are not difficult to derive following the discussion in Appendix A of \cite{lin2019}.

\begin{subsection}{BB84 with Phase Error Parameter Estimation}\label{subsec:PhaseBB84}
As a simple case where the analytic answer is known, we consider the BB84 protocol where signals are sent in the $Z$-basis with probability $p_{z}$ and the key map is only done on the $Z$-basis so that all other events are removed during generalized sifting \cite{lo00suba}. In other words, the states $\ket{0}, \ket{1}$ are sent with probability $\frac{p_z}{2}$ and the states $\ket{+},\ket{-}$ are sent with probability $\frac{1-p_z}{2}$. We assume that Alice and Bob perform parameter estimation in which they only check the phase error, $e_{x}$, using the POVM 
\begin{align} \label{eq:phaseErrorPOVM}
    \{\Pi_{e_{x}},\bbone - \Pi_{e_{x}}\}
\end{align} 
where $\Pi_{e_{x}} \equiv (\bbone_{A} \otimes \bbone_{B} - \sigma_{X} \otimes \sigma_{X})/2$ and $\sigma_{X}$ is the Pauli-$X$ operator. An analytic key length in finite size for a protocol with unique acceptance has been given for this scenario in \cite{scarani2008}:\footnote{In \cite{scarani2008} the authors do not have a factor of half for the variation bound $\mu$ in their corresponding formula for $H_{\mu}(X|E)$ (Eqn. 6 of \cite{scarani2008}). However, one can tighten their result as if one is to perturb a probability distribution of two outcomes by a total amount $\mu$ and maintain a probability distribution, the most one can increase one outcome by is $\mu/2$.}
\begin{align}\label{eq:phaseErrKeyRate}
    \ell_{\text{BB84},e_{x}} =& n [1-h(e_{x} + \frac{\mu}{2})-f_{\EC}h(e_{z})-\delta(\bar{\varepsilon})] \\
    & \hspace{.75cm} -\log_{2}(\frac{2}{\varepsilon_{\EC}}) - 2 \nonumber \log_{2}(\frac{2}{\varepsilon_{\PA}})
\end{align}
where $H_{\mu}(X|E) = 1-h(e_{x} + \frac{\mu}{2})$, $H(X|Y) = h(e_{z})$, $h(p) = -p \log_{2}(p) - (1-p) \log_{2}(1-p)$ , and all other terms are as defined in Section \ref{sec:background}. The specific form of $H_{\mu}(X|E)$ follows from the fact that asymptotically the conditional entropy between the key and Eve for this protocol is of the form $1-h(e_{x})$ and so the worst case-scenario in the finite regime is that half of the variation bound $\mu$ increases the phase error. The error correction term $H(X|Y)$ is the binary entropy of the quantum bit error rate $e_z$ because the number of bits that needed to be corrected when doing the key map from the Z-basis is the error rate in the Z-basis.

To show that our approach works, we consider it against the analytical curve in Fig. \ref{fig:compVTheory}. Following \cite{scarani2008}, to determine the value of $\mu$, we assume Alice and Bob sacrifice $(1-p_z)^2 N$ of the signals to parameter estimation. This is a good choice since because as $N$ goes to infinity, $p_{z}$ can approach zero, and so one will need to spend a continuously smaller fraction on parameter estimation, which this a priori decision takes into account. Furthermore, in the simulation we assume that our observations yield that the error rates satisfy $e_{z} = e_{x}$ to let $H(X|Y) = h(e_x)$. As can be seen in Fig. \ref{fig:compVTheory}, for this protocol our solver produces a lower bound that matches the analytical result perfectly. Furthermore, in this example, we let $f_{\EC} = 1.2$ as this is a realistic model of the inefficiency of error correction in current experiments \cite{scarani08b,scarani2008,Lucamarini2015}.

\begin{figure}[h]
    \centering
    \includegraphics[width=0.9\columnwidth]{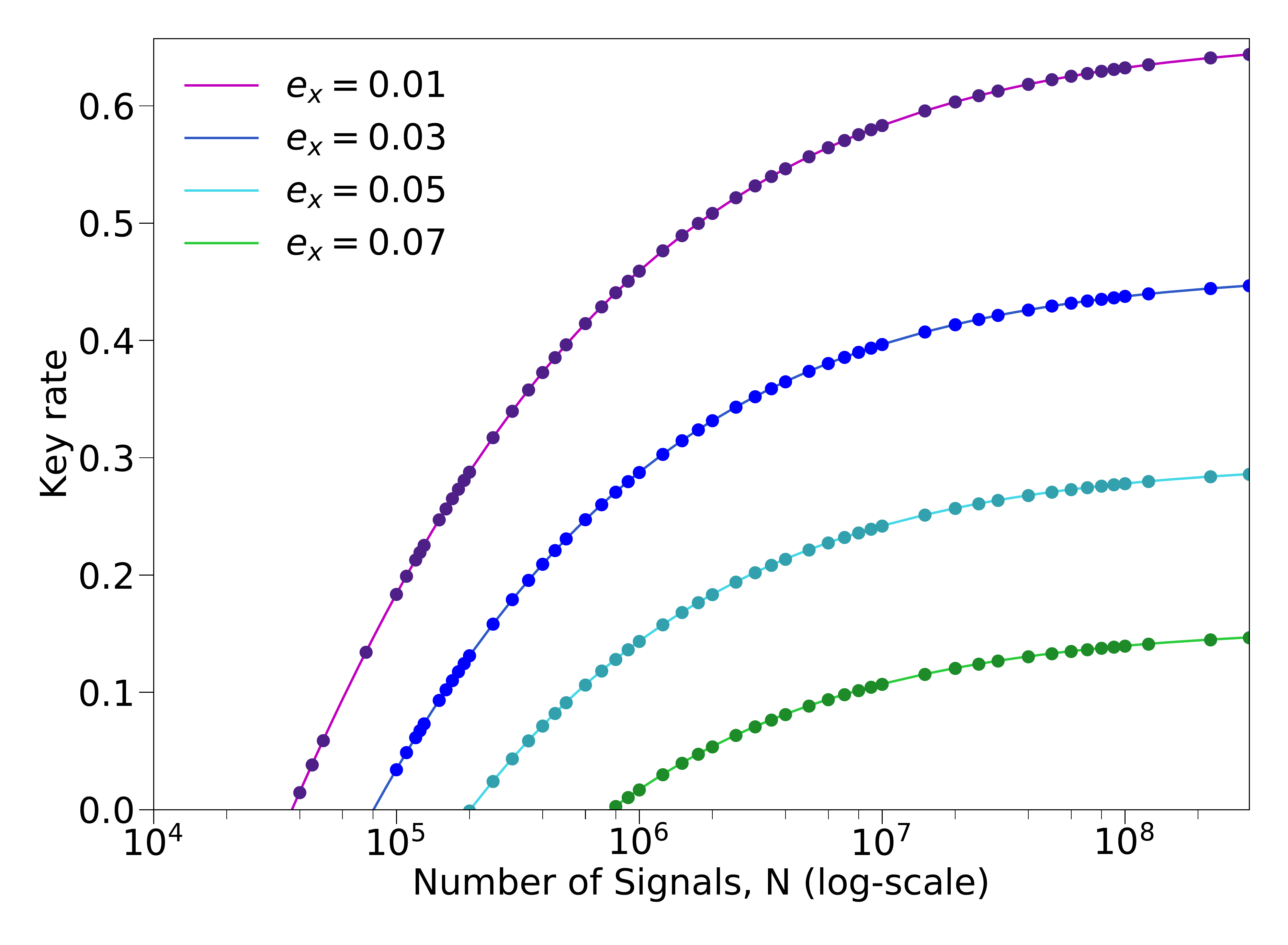}
    \caption{Numerical key rate versus analytic key rate for BB84 for four error rates with $\varepsilon_{\PE} = \bar{\varepsilon} = \varepsilon_{\EC} = \varepsilon_{\PA} = \frac{1}{4}\times 10^{-8}$ so that $\varepsilon = 10^{-8}$. The lines are the theory curves, and the dots are the corresponding solutions by our numerical method. We let $e_{z} = e_{x}$ and $p_z = 0.9$. We assume here that the sample size is still larger than the block length of error correction, which then gives $f_{\EC}$ =1.20.}
    \label{fig:compVTheory}
\end{figure}
\end{subsection}

\begin{subsection}{Rotated BB84 \& POVM Choice}\label{subsec:rotatedBB84}
In this section we explore the effect of fine-grained data versus coarse-grained data on the key rate and the increased importance of the difference in the finite regime. Furthermore we show the advantage of considering multiple coarse-grainings (Eqn. \ref{eq:paramEstSetGeneral}) rather than only one (Eqn. \ref{eq:paramEstSet1}). This in turn shows that a major advantage of our numerical method is the ability to consider multiple coarse-grainings to achieve tight key rates which analytically is not manageable.

In the case of constraining the set of density matrices using a single frequency distribution $F$, there are two competing effects--- the rate at which the variation bound $\mu$ goes to $0$ and the value of the asymptotic key rate. As one can see from Eqn. \ref{eq:mu}, the number of POVM outcomes effects the size of the variation bound $\mu$. This means that more coarse-grained data $F^{C_{k}}$ has a variation bound $\mu_{k}$ that converges to $0$ faster than that of the fine-grained data. It follows that for a case such as in the first example where an element of a coarse-grained probability distribution ($e_{x}$) determines the key rate (Eqn. \ref{eq:phaseErrKeyRate}), the coarse-grained data will lead to a better or equal key rate to the fine-grained data for any amount of signals.

However, we know that if one applies a unitary rotation about the Y-axis on the Bloch sphere to each signal sent to Bob, then the fine-grained statistics will detect the rotation, thereby leaving the key rate unchanged. In contrast, the phase error coarse-grained statistics cannot determine the rotation, thereby decreasing the coarse-grained key rate. As asymptotically the fine-grained key rate is better than the coarse-grained key rate in the event of such a rotation, even with the coarse-grained statistic variation bound converging to zero faster, the fine-grained key rate must be better than the coarse-grained key rate for some number of signals. 

Independent of finite size effects, the idea that some POVMs being robust to rotations has already been recognized in the literature by the invention of the `reference frame independent' and `6-state 4-state' protocols  \cite{laing2010,tannous2019}. The idea is that the information extracted by the POVM determines how robust the protocol is to differences in Alice and Bob's reference frames. This is because the signals sacrificed for the parameter estimation step allow them to in effect align their relevant reference frame  \cite{bartlett2007}. For example, if we had rotated the states about the $X$-axis of the Bloch sphere, not even the fine-graiend data of the BB84 protocol would help, but the six-state protocol, which is tomographically complete, would be robust to this. In this section we present an example of this misalignment in reference frames in BB84 to explore its relation to finite size effects and the advantage of doing parameter estimation with multiple-coarse grainings.

We consider BB84 where we constrain with one or more of the following three conditional probability distributions where for intelligibility we write the corresponding POVM rather than the conditional probability distribution:
\begin{enumerate}
    \item The fine-grained joint POVM constructed by both Alice and Bob having the local POVM:
    \begin{align}\label{eq:fineGrainedBB84}
    \left \{p_{z} \ket{0}\bra{0}, p_z \ket{1}\bra{1}, (1-p_{z}) \ket{+}\bra{+}, (1-p_{z}) \ket{-}\bra{-} \right \}
    \end{align}
This corresponds to applying the identity conditional probability distribution to the fine-grained statistics.
    \item The phase error POVM defined in Eqn. \ref{eq:phaseErrorPOVM}. This corresponds to mapping the frequencies corresponding to Alice and Bob both using the $X$-basis POVM and getting different results to a single outcome and all other fine-grained outcomes to a second.
    \item The \textit{agreement} POVM which simply checks how often Alice and Bob agree:
     $$\{p_z^2 \Pi_{0}, p_z^2 \Pi_{1}, \frac{(1-p_z)^2}{2}\Pi_{+}, \frac{(1-p_z)^2}{2}\Pi_{-}, \Pi_{else} \} $$
     where $\Pi_{a} = \ket{a}\bra{a} \otimes \ket{a}\bra{a}$ and $\Pi_{else}$ is the POVM element that completes the POVM. This corresponds to a conditional probability distribution that retains the statistics pertaining to Alice and Bob getting the same outcome and mapping all other fine-grained outcomes to a single outcome.
\end{enumerate}
To evaluate the resulting key rates, we need to work with simulated observations, as we do not work from actual experimental data. To simulate the observed statistics, we consider a simple noise model for a qubit channel. Alice sends half of the ideal state $\ket{\Phi^{+}} \equiv \frac{1}{\sqrt{2}}(\ket{00}+\ket{11})$ through a channel. The channel is the composition of two channels. The first channel is the depolarizing channel with noise value $q$ defined as:
\begin{align*}
    \Phi_{dp}^{q}(X) &= \sum_{k=0}^{3} p_k \sigma_k (X) \sigma_k
\end{align*}
where $p_0 = 1-\frac{3q}{4}, p_1 = p_2 = p_3 = \frac{q}{4}$ and $\sigma_0 = \bbone_{2}, \sigma_1=\sigma_X, \sigma_2=\sigma_{\text{Y}}, \sigma_3=\sigma_{Z}$ where $\sigma_{X}, \sigma_{\text{Y}}, \sigma_{Z}$ are the Pauli operators. The depolarizing channel induces a qubit error rate of $q$ in the output state. The second channel is a unitary channel that rotates the state about the Y-axis on the Bloch sphere by an angle $\theta$, $\Phi_{U}(X) = e^{i \theta \sigma_{\text{Y}}}Xe^{-i \theta \sigma_{\text{Y}}} $. Alice and Bob then perform measurements on the state $(\mathcal{I}_{A} \otimes (\Phi_{U} \circ \Phi^{q}_{dp}))(\ket{\Phi^{+}}\bra{\Phi^{+}})$ using one of the POVMs previously described to generate the probabilities.

\begin{figure}[h]
    \centering
    \includegraphics[width=\columnwidth]{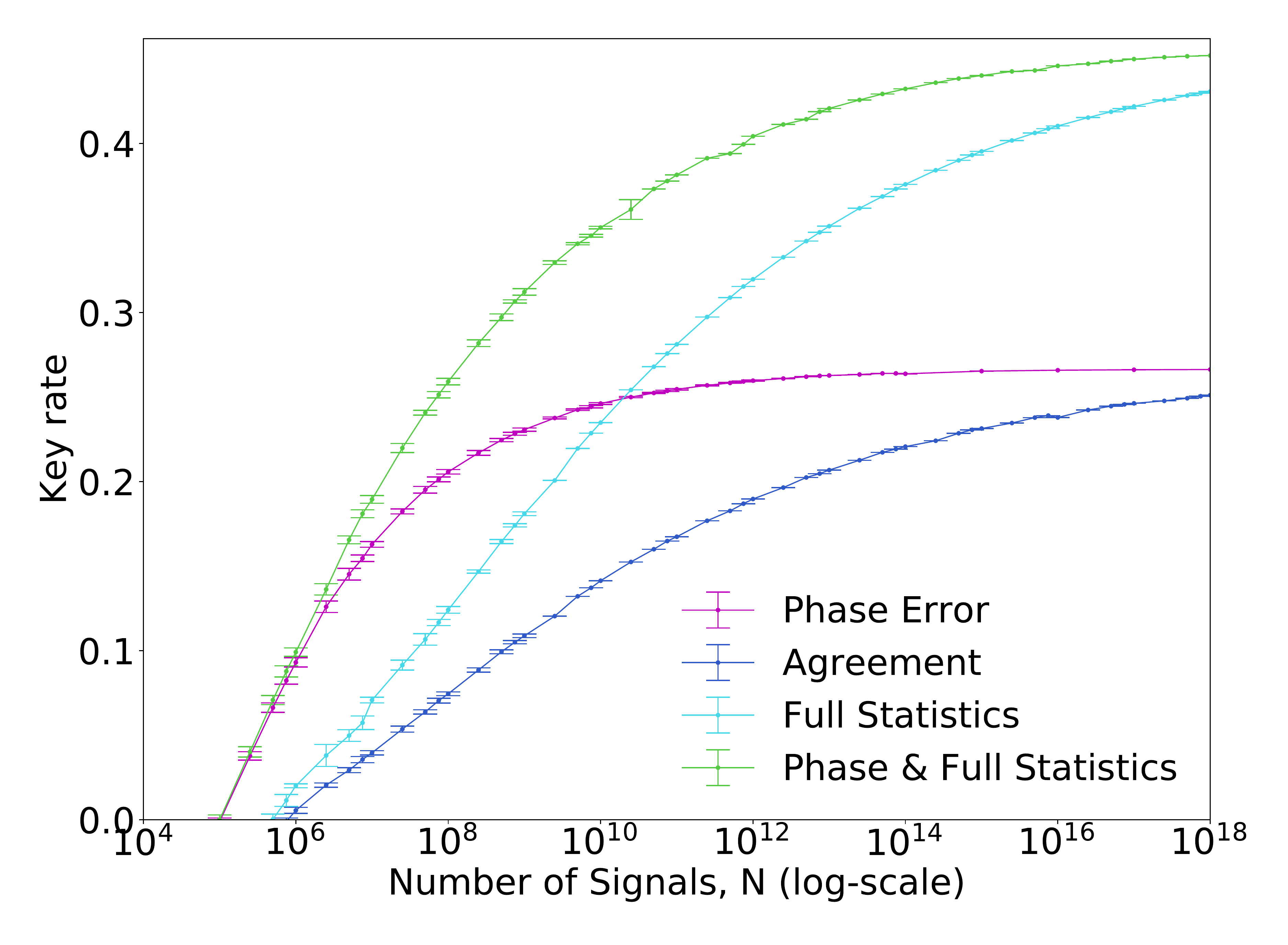}
    \caption{We consider four different parameter estimation constraints for $BB84$ transmitted through a depolarizing channel with $q = 0.02$ when the signal states have been rotated by $12^{\circ}$ about the $Y$-axis on the Bloch sphere. Each point has $p_{z}$ numerically optimized for maximum key rate. The error bars are from checking the key rate for 20 trials of sampling the distribution whenever the number of signals used for parameter estimation was less than $10^{8}$ and calculating the standard deviation. Note that the phase error curve is what is achievable using previous methods for finite key analysis which have restrictive assumptions, whereas the other curves which can improve the key rate significantly in certain regimes are achieved through our numerical method's ability to consider multiple coarse-grainings and POVMs with more than two outcomes. For all curves we let $\varepsilon_{\PE} = \bar{\varepsilon} = \varepsilon_{\EC} = \varepsilon_{\PA} = \frac{1}{4}\times 10^{-8}$.}
    \label{fig:rotation}
\end{figure}   

In Fig. \ref{fig:rotation}, we plot the key rate for all three coarse-grainings individually as well as the key rate when we consider both the phase-error statistics and the fine-grained statistics. To look at this, in Fig. \ref{fig:rotation}, whenever $m \leq 10^{8}$ we construct a frequency distribution by randomly sampling the simulated probability distribution using a pseudo-random function and then calculate the key rate for the protocol with unique acceptance which accepts on that frequency distribution. To see how much the key rate fluctuates when sampling $m$ times depending on the frequency distribution Alice and Bob accept, we chose to repeat the simulation 20 times to determine the average key rate and standard deviation of the protocol with unique acceptance with all other parameters fixed. The standard deviation is represented by the error bars in Fig. \ref{fig:rotation}. Furthermore, to make the comparison between the different POVMs fair, we optimize the choice of $p_{z}$ at each point by maximizing the average key rate over $p_{z}$ given that specific value of $N$. As in the previous example, we let $m = (1-p_{z})^2 N$ and assume they do the key map only in the $Z$-basis. Lastly, the (observed) error correction cost for all four key rates is $f_{\EC} H(\text{X}|\text{Y}) = f_{\EC} h(\bar{e}_{z})$ where $f_{\EC} = 1.2$ and $\bar{e}_{z}$ is the bit error frequency determined by the fine-grained statistics in the key-generation basis Z.

Given Fig. \ref{fig:rotation}, we now see how in some regime coarse-graining does better than fine-grained data in some regime due to the coarse-grained variational bound $\mu_{k}$ converging to zero faster, but is ultimately worse as $N$ increases because asymptotically the fine-grained data provides a better key rate. We also see that considering both frequency distributions improves the key rate for all $N$. This is because whatever density matrix satisfies both sets of constraints has the phase error lower than just the fine-grained data and the unitary is `undone' to a greater degree than just the phase error coarse-grained data. For this reason in the finite regime it will only be beneficial to always optimize over the fine-grained data as well as relevant coarse-grainings. The ability for our solver to do this regardless of the number of outcomes is one property which makes our solver truly general.

\end{subsection}

\begin{subsection}{MDI-BB84 with Qubits}\label{subsec:MDIBB84}
In this section we show that our numerical method can be extended to MDI-QKD protocols which are designed to be immune to side-channel attacks on measurement devices \cite{lo2012}. Specifically we consider MDI-BB84 with perfect single photon sources in which Alice and Bob both send BB84 states to an untrusted third party Charlie who performs Bell state measurements on the two signals. Charlie then announces on which signals his measurement was successful as well as the outcome. Alice and Bob then do sifting on this subset and finally construct the key. The primary extension for finite key is that in MDI-QKD there is a third party. This means that there is a joint probability distribution over three alphabets and a joint POVM over three parties. This however is an immediate extension as parameter estimation can be defined for tripartite states and the third party in MDI QKD is a classical announcement and so does not effect Alice and Bob's fine-grained data.

To simulate data for the protocol, we apply source-replacement to both Alice's and Bob's signal states resulting in a state $\rho_{ABA'B'}$. In our calculation, we assume the setup is using linear optics, so Charlie can only discriminate unambiguously two of the Bell state measurements, $\Psi_{+}$ and $\Psi_{-}$ where $\Psi_{\pm} \equiv \frac{1}{\sqrt{2}}(\ket{01} \pm \ket{10})$. For simulating the statistics, we consider that the signal portions of the states, $A'$ and $B'$, each go through a separate depolarizing channel $\Phi_{dp}^{q}$ as they are sent to Charlie. Lastly, we assume Alice and Bob only do the key map in the $Z$-basis for simplicity. In Fig. \ref{fig:MDIBB84} we consider MDI-BB84 with $p_z =0.5$ for two depolarizing parameter values to see the rate of converging to the asymptotic key rate as a simple example
\end{subsection}

\begin{figure}[h]
    \centering
    \includegraphics[width=0.9\columnwidth]{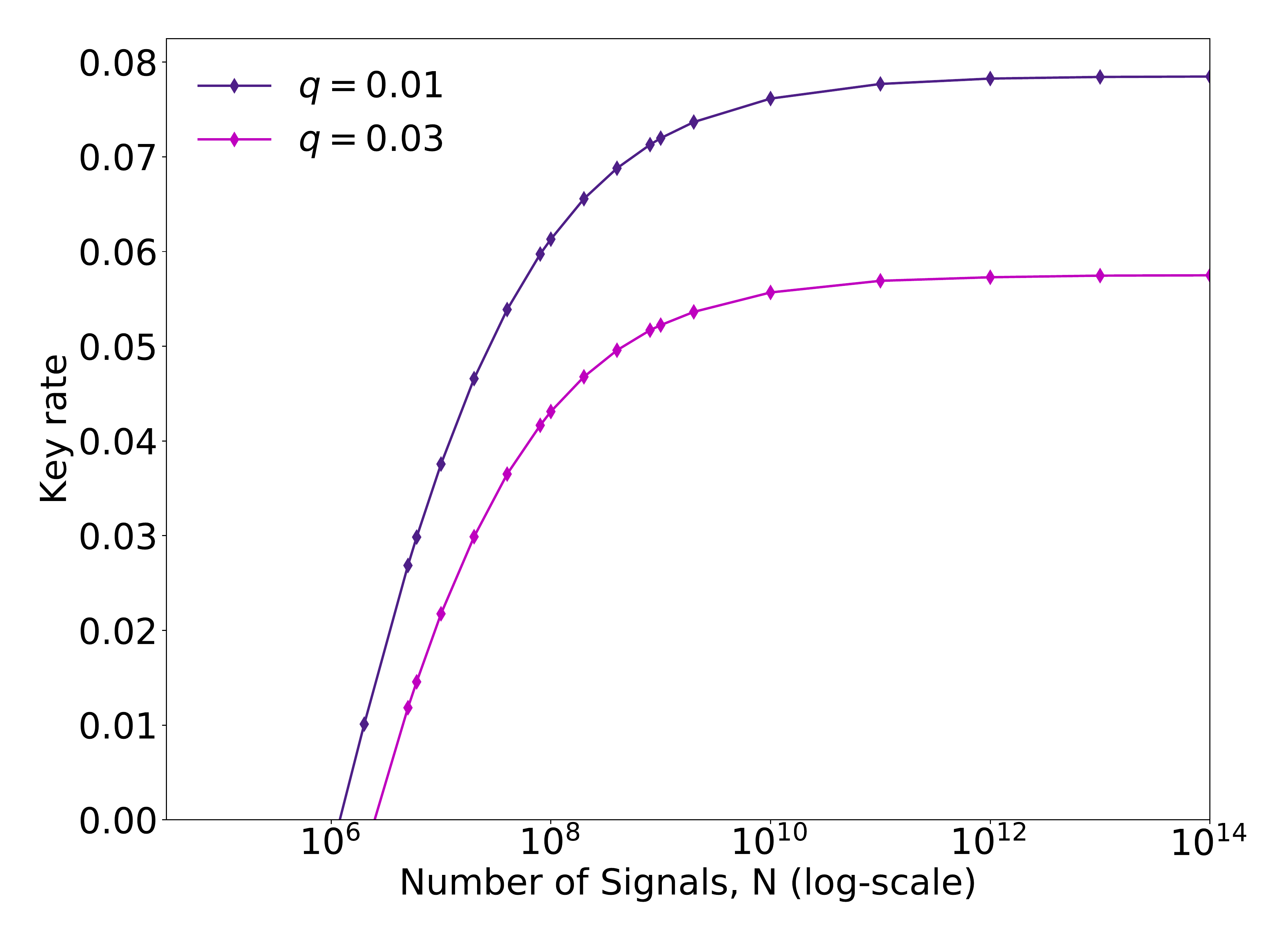}
    \caption{Here we see the MDI-BB84 protocol with unique acceptance converging to its asymptotic value as the number of signals is increased for depolarizing channels with depolarizing parameter values $q = 0.01$ and $q = 0.03$. For all curves, the security is defined by $\varepsilon_{\PE} = \bar{\varepsilon} = \varepsilon_{\EC} = \varepsilon_{\PA} = \frac{1}{4}\times 10^{-8}$.}
    \label{fig:MDIBB84}
\end{figure}

\begin{subsection}{\prot}\label{subsec:DPRBB84}
We next apply our method to optical implementation of QKD protocols with weak coherent pulses. Since each state that Bob receives is an optical mode and is in principle manipulated by Eve, a full description of the POVM usually involves an infinite-dimensional Hilbert space (e.g. Fock space). This also means that the density operator $\rho_{AB}$ in our optimization problem is infinite-dimensional such that no numerical optimization algorithm can solve the problem directly. Fortunately, for many discrete-variable QKD protocols, there exists a squashing model \cite{beaudry08a,moroder10a,Zhang2020} that reduces the apparent infinite-dimensional representation to an effective finite-dimensional subspace representation. This shows that our numerical method applies even for optical implementations so long as they can be represented in finite-dimensional Hilbert spaces. Here, we present our finite key analysis for the discrete-phase-randomized BB84 protocol \cite{cao2015}, which is based on phase-encoding and has a squashing model \cite{beaudry08a}.

We consider the following simple model for determining the statistics. 
\begin{figure}[h]
    \centering
    \includegraphics[width =0.9\columnwidth]{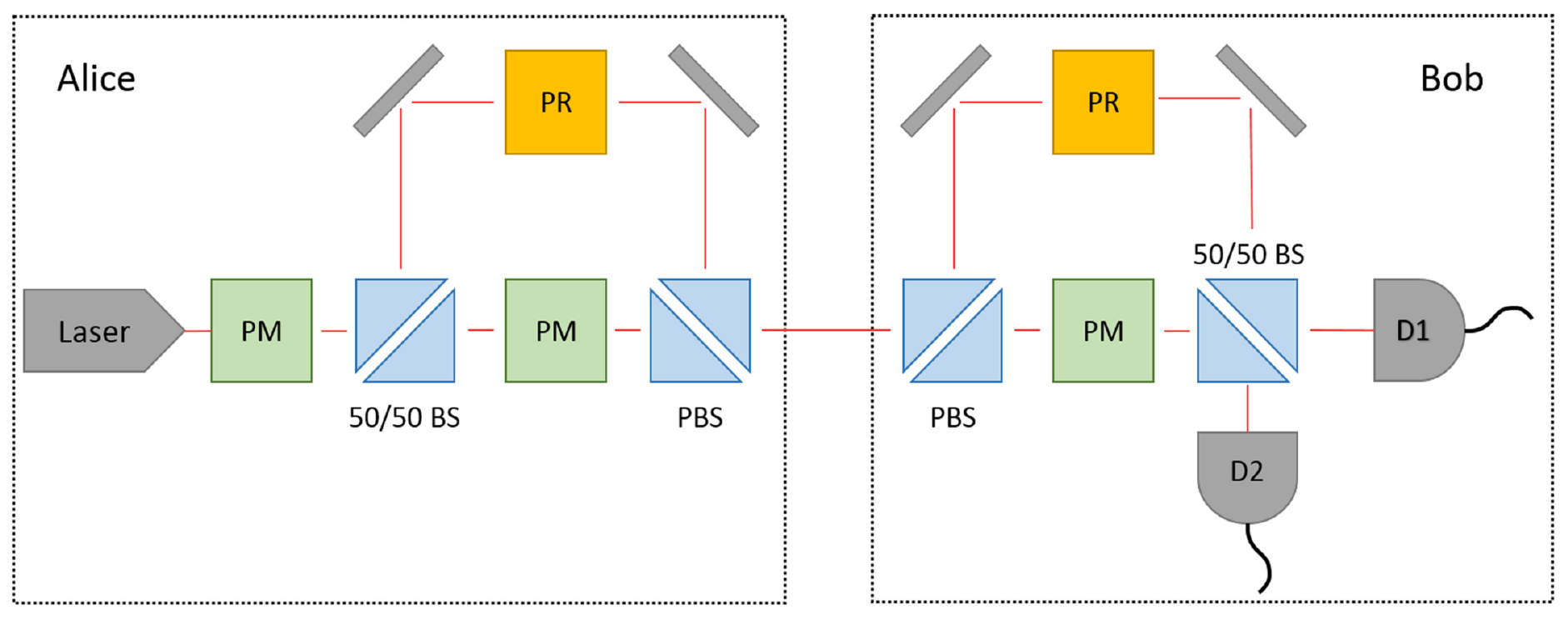}
    \caption{Schematic for discrete-phase-randomized BB84. PM stands for phase modulator, PBS stands for polarizing beam splitter, BS stands for beam splitter, PR stands for polarization rotator, and D1 and D2 are two threshold detectors.}
    \label{fig:protocol-diagram}
\end{figure}

As depicted in Fig. \ref{fig:protocol-diagram}, the quantum part of the protocol is 
\begin{enumerate}
    \item Alice sends two-mode coherent states $\ket{\sqrt{\nu} e^{i\theta}}_r \ket{\sqrt{\nu} e^{i(\theta + \phi_{A})}}_s$, to Bob where the first mode is the reference pulse and the second mode is the signal pulse. The global phase $\theta$ is chosen at random from the set $\{\frac{2\pi k}{c}: k=0, \dots, c-1\}$ where $c$ is the number of different global phases. The key information is encoded in the relative phase $\phi_A$ chosen from the Z basis $\{0, \pi\}$ or X basis $\{\frac{\pi}{2}, \frac{3\pi}{2}\}$.
    \item After receiving states from Alice, Bob may choose to measure in one of the two basis by applying a relative phase $\phi_B \in \{0, \frac{\pi}{2}\}$ to the reference pulse, where $\phi_B = 0$ corresponds to Z basis and $\phi_B = \frac{\pi}{2}$ to X basis. This results in either one, none, or both of Bob's detectors clicking. In the case where both detectors click, Bob assigns the result to either just detector 1 clicking or just detector 2 clicking.
\end{enumerate}

We remark that the protocol with c=1, in which case Alice does not randomize the global phase, is also studied in \cite{lo2006,lin17}.

For our simulation, we consider a lossy channel parameterized by the single-photon transmittance $\eta = 10^{-\alpha_{att} L /10}$ for a distance $L$ (in kilometers) between Alice and Bob. We also introduce a channel noise parameterized by $\zeta$, which describes the relative phase drift between the signal pulse and the reference pulse. In addition, imperfection of Bob's detectors is taken into account by the dark count probability $p_d$ and the dectector efficiency $\eta_d$. To obtain simulated statistics, we choose $\eta_{d} = 0.045$, $p_d = 8.5 \times 10^{-7}$, and let the attenuation coefficient be $\alpha_{att} = 0.2$ dB/km, from the experimental parameters reported in \cite{gobby04a}. We also set $\zeta = 11^{\circ}$, which produces a misalignment error of 1$\%$ at 0 km distance and let $f_{EC} = 1.16$ as was done in \cite{lin17}.

Under the squashing model and source-replacement scheme, the fine-grained statistics for this protocol are generated by a $20c$-outcome joint POVM constructed by Alice and Bob's local POVMs where Alice has $4c$ POVM elements which are projectors on to her $4c$ possible signal states and Bob has a 5-outcome POVM defined as:
\begin{align*}
    \{ & 1/2 \ket{0}\bra{0} \oplus 0, 1/2 \ket{1}\bra{1} \oplus 0, \\
     & 1/2 \ket{+}\bra{+} \oplus 0, 1/2 \ket{-}\bra{-} \oplus 0, \ket{\text{vac}}\bra{\text{vac}} \}
\end{align*}
In other words, Bob's local POVM is the standard fine-grained local BB84 POVM (Eqn. \ref{eq:fineGrainedBB84} with $p_{z} = 1/2$) embedded in a three-dimensional space plus a projector onto the third dimension where the third dimension is the vacuum state and $\ket{\text{vac}}$ denotes the basis of the third dimension.

We take $L = 100$ km and $L = 20$ km and consider both $c=1$ and $c=2$ scenarios as an example to show the method works for multiple discrete phases and loss regimes. In this model the dark counts are the primary source of error. In generating this plot, to improve the key rate when less signals are sent, we optimize the fraction of signals that would be used for parameter estimation, which we denote $g_{\PE} \equiv m/N$, heuristically. The fraction is determined as follows:
\begin{align*}
    g^{L=20km}_{\PE} =  
    \begin{cases}
        0.99 & N < 1.31 \times 10^{11} \\
        \frac{1.1 \times 10^{11}}{N} + (0.5)^{\log_{10}(N)/4} & \text{else} \\
    \end{cases}
\end{align*}
\begin{align*}
    g^{L=100km}_{\PE} =  
    \begin{cases}
        0.99 & N < 2.75 \times 10^{14} \\
        \frac{2.35 \times 10^{14}}{N} + (0.5)^{\log_{10}(N)/5} & \text{else} \\
    \end{cases}
\end{align*}
The first term of line 2 of each $g_{\PE}$ was determined by numerically determining for how many signals the key rate could be made positive for $c=1$. The extra term was decided so as to sacrifice a smaller fraction to parameter estimation as $N$ grows so that the key rate is improved.

\begin{figure}
    \centering
    \includegraphics[width = 0.9 \columnwidth]{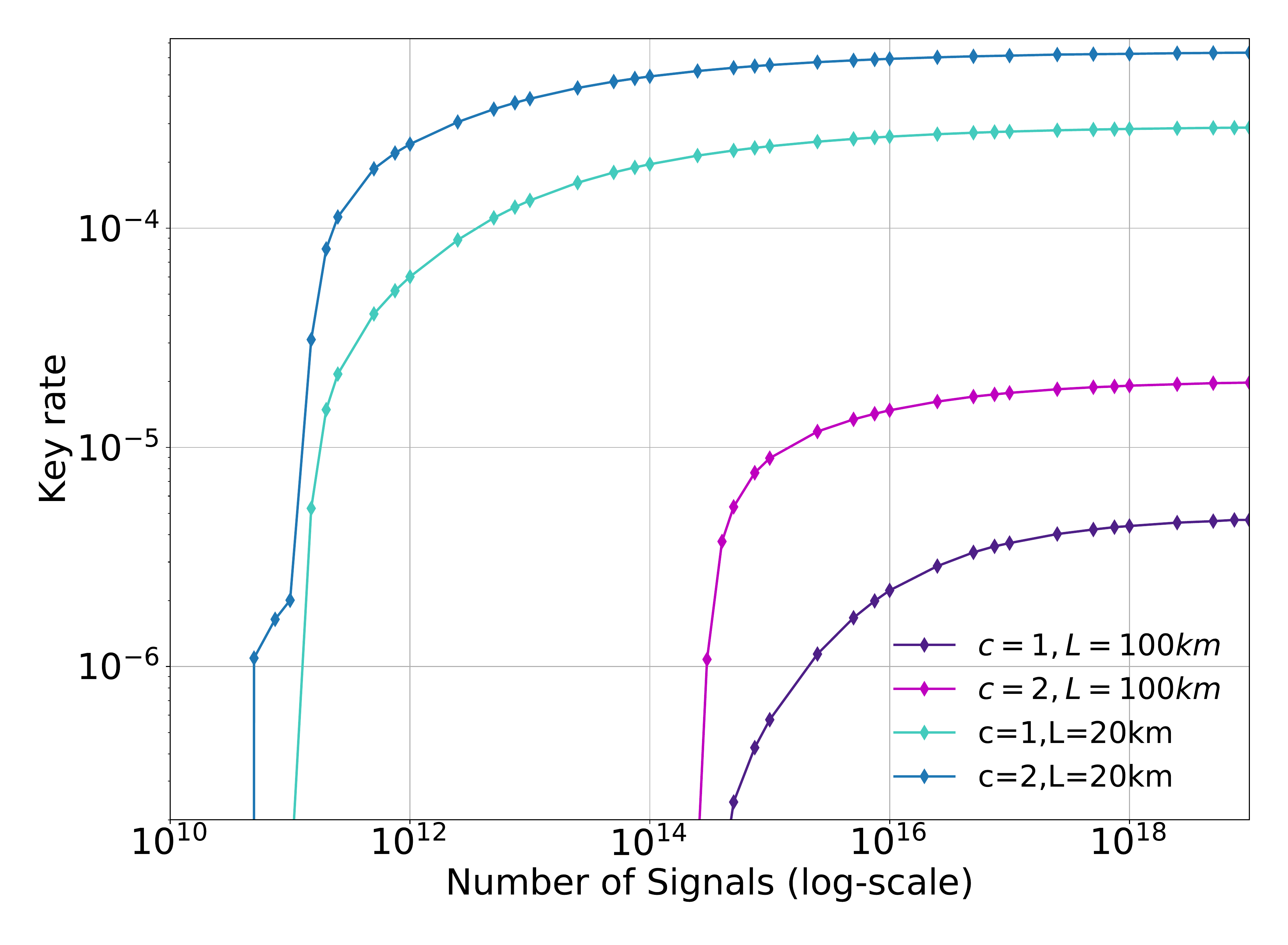}
    \caption{Key rate of discrete-phase-randomized BB84 with unique acceptance when not randomizing the global phase ($c=1$) and randomizing it over 2 choices ($c=2$). Every point is for optimized coherent state intensity $\nu$. For all curves, the security is defined by $\varepsilon_{\PE} = \bar{\varepsilon} = \varepsilon_{\EC} = \varepsilon_{\PA} = \frac{1}{4}\times 10^{-8}$. For this protocol we let $f_{\EC} = 1.16$.}
    \label{fig:protFig}
\end{figure}

We notice that with our simulation parameters, at $L=100$ km considered in Fig. \ref{fig:protFig}, a significant amount of signals needs to be sent before the key rate becomes nonzero. The reason is that at $L = 100$ km, the probability of the outcomes that will lead to key generation is quite low, at the order $10^{-6}$ in the $c=1$ case. It follows that if the variation bound $\mu$ is of an order greater than $10^{-6}$, there exists a probability distribution $P$ such that $\|P-F\|_{1} \leq \mu$ and $P$ corresponds to a density matrix that lacks sufficient correlation for any key to be distilled. Therefore one needs to sacrifice enough signals to parameter estimation such that the variation bound $\mu$ is sufficiently small with respect to the portion of the frequency distribution relevant to key distillation. 

\end{subsection}

\begin{subsection}{Security of BB84 with Practical Acceptance Probability}

\begin{figure}[h]

\subfloat[]{\label{fig:fullSecPhaseSet}\includegraphics[width=\linewidth]{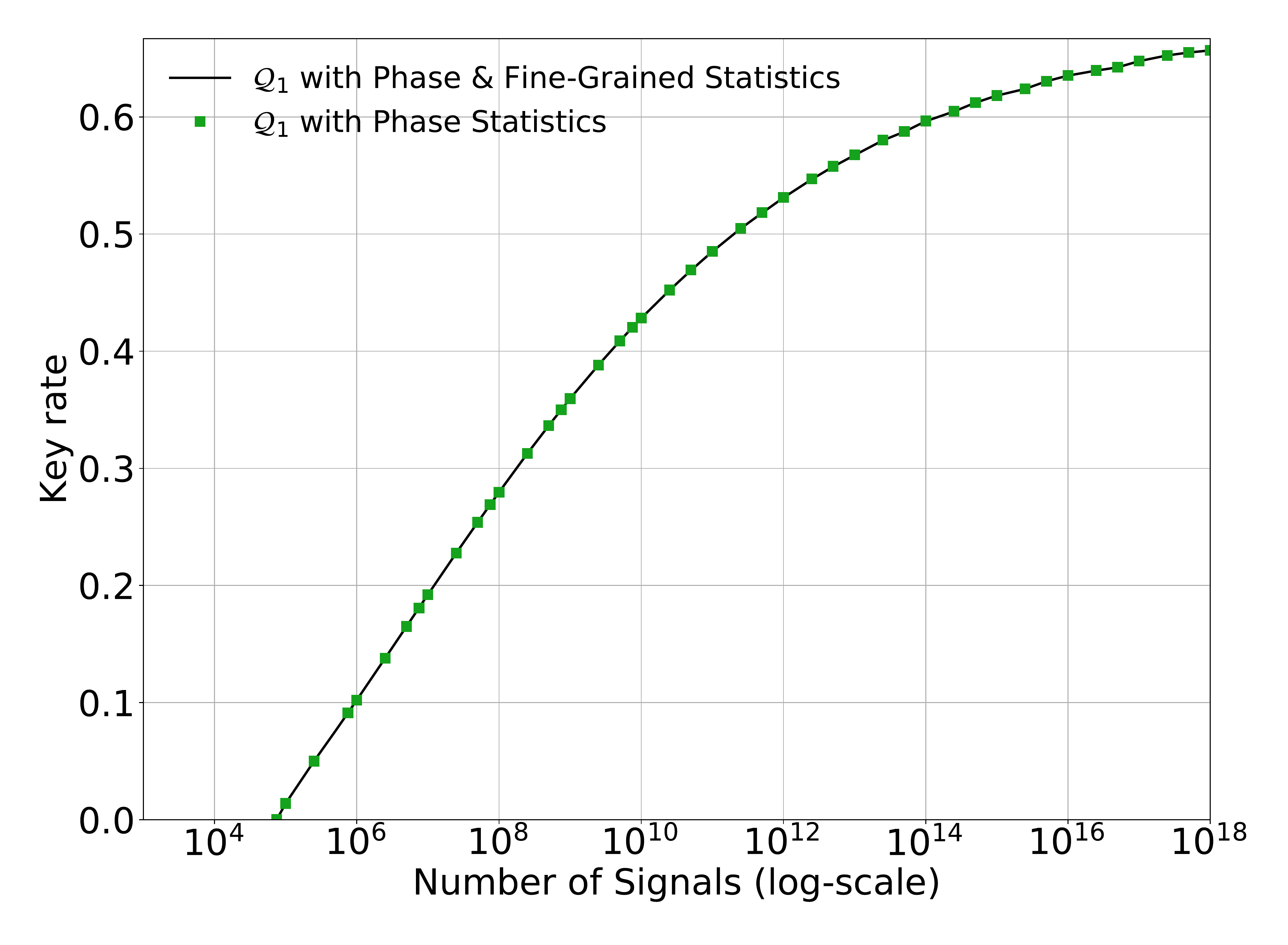}}\\

\subfloat[]{\label{fig:fullSecRotatedSet}\includegraphics[width=\linewidth]{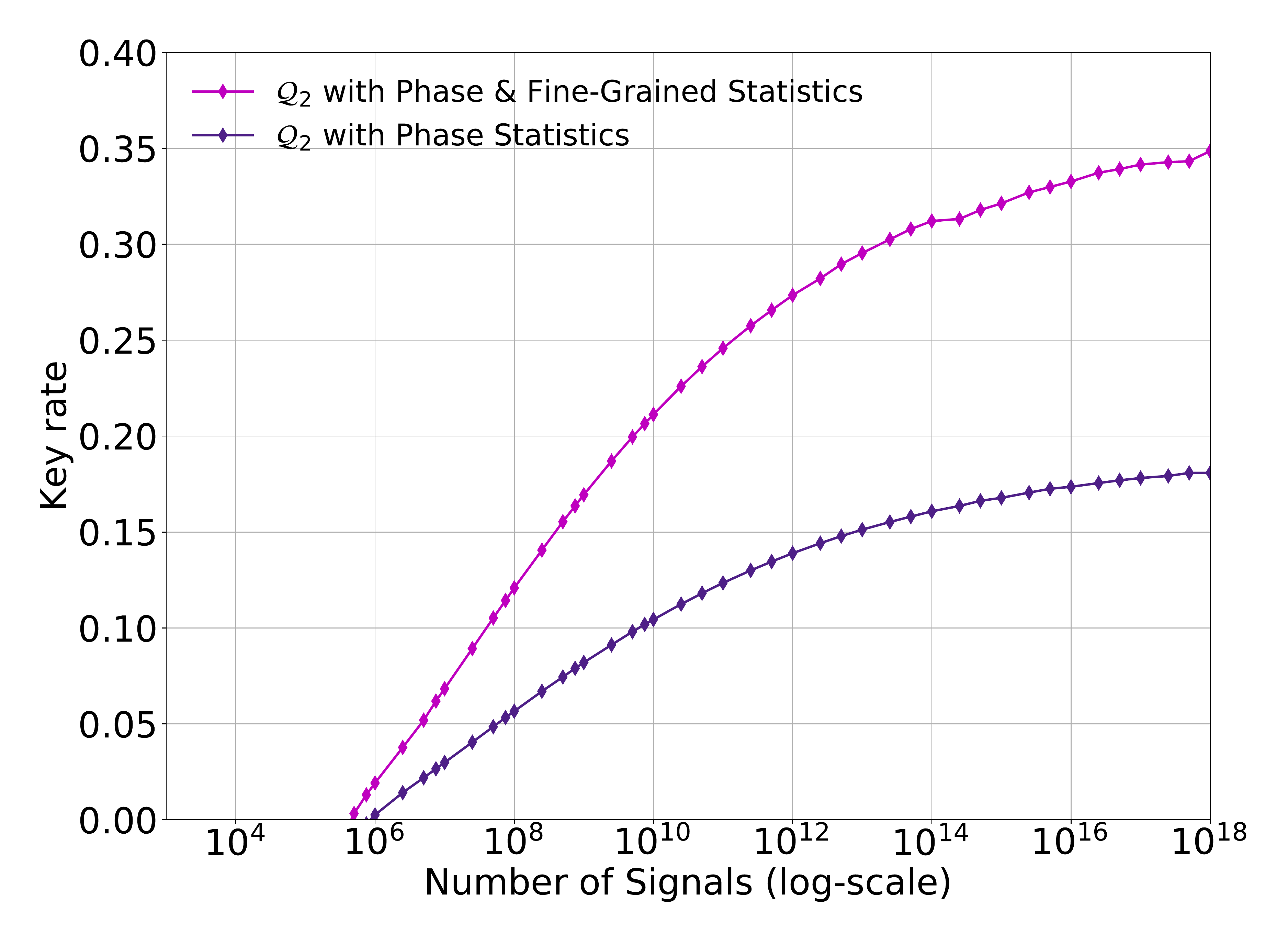}}

\caption{\label{fig:fullSecTotalFig} (a) Key rate of the BB84 protocol for accepting statistics in $\mathcal{Q}_{1}$ where either just the phase statistics or both the phase statistics and the fine-grained statistics are used to determine the key rate. We see in this case the fine-grained data does not help for this protocol. (b) Key rate of the BB84 protocol for accepting observed statistics in $\mathcal{Q}_{2}$ where either just the phase statistics or both the phase statistics and the fine-grained statistics are used to determine the key rate. We see in this case the fine-grained data help for this protocol, and so an analytical key rate calculation is difficult.
For both (a) and (b), each point $p_{z}$ is optimized and the security is defined by $\varepsilon_{\PE} = \bar{\varepsilon} = \varepsilon_{\EC} = \varepsilon_{\PA} = \frac{1}{4}\times 10^{-8}$.}

\end{figure}

So far we have only presented protocols with unique acceptance. However, protocols with unique acceptance are impractical as the probability that an experiment yields the exact frequency distribution of outcomes that match the acceptance criteria is usually very low. Thus one introduces a range of accepted statistics, where the key rate is now to be taken over the worst case scenario of the accepted statistics. Therefore, there is a trade-off between how often one aborts, and the length of the secret key generated when the protocol does not abort. In some cases, especially where the accepted statistics is only based on one observable, such as an error rate, and the key rate has some monotonic behaviour, it is easy to identify the worst-case acceptable statistics. In these cases one can relate the case of a set of accepted statistics back to the case of a single accepted statistics, namely the identified worst case statistics. However, in cases where the observed statistics needed for determining the key rate of the protocol are more complex, it is often not as simple to identify the worst case statistics. In these scenarios, our numerical method is a powerful tool for determining a tight lower bound of the secret key rate. Here we present an example of determining the secure key rate for single-photon BB84 in the practical setting where multiple frequency distributions are accepted by Alice and Bob to show how our numerical approach may help.

 We again return to the case where Alice and Bob perform BB84 where they choose some probability $p_{z}$ to send signals in the $Z$-basis. We consider two sets of frequency distributions to accept corresponding to whether their protocol has ideal behaviour or is suffering from misalignment due to the quantum channel. Following the notation in Eqn. \ref{eq:SmuMainText}, the first set, $\mathcal{Q}_{1}$, is defined by letting $\overline{\mathcal{N}}(\overline{F})$ be the two-outcome frequency distribution of `phase error' and `no phase error' with no observed phase errors ($e_{x} = 0$). We refer to $\mathcal{Q}_{1}$ as the phase set. The second set, $\mathcal{Q}_{2}$, is defined by letting $\overline{\mathcal{N}}(\overline{F}) = \overline{F}$ be the asymptotic results of the fine-grained statistics given the model from Section \ref{subsec:rotatedBB84}. We refer to $\mathcal{Q}_{2}$ as the rotated set. In both cases, the variation threshold, $t$, is $2 (1 - p_{z})^{2} \overline{e}_{x}$ where $\overline{e}_{x}$ is the maximum tolerated observed error from $\overline{F}$. For this example we let $\overline{e}_{x} = 0.02$. The factor of $(1-p_{z})^{2}$ is so that the variation threshold stays the same as $p_{z}$ is varied to optimize the key rate. 

Given the definition of the phase set, $\mathcal{Q}_{1}$, the key rate can be determined analytically as one can replace $e_{x}$ in Eqn. \ref{eq:phaseErrKeyRate} by $\overline{e}_{x}$. Furthermore, as no data more fine-grained than the phase error is needed in this case, it is clear that the multiple coarse-grainings will not further improve the key rate. These observations are verified numerically in Fig. \ref{fig:fullSecPhaseSet}. However, in the case where the observed statistics would be contained in $\mathcal{Q}_{2}$ rather than in $\mathcal{Q}_{1}$, an analytical tight lower bound of the key rate is not a reasonable task as the structure of the worst case scenario is no longer simple. This is seen in Fig. \ref{fig:fullSecRotatedSet}, where our numerical result shows that multiple coarse-grainings helps to obtain a tighter key rate when $\mathcal{Q}_{2}$ is used. It follows that obtaining a tight key rate analytically would be difficult as one needs to utilize both fine-grained statistics and multiple coarse-grainings. 

More generally, this tells us the optimal choice of $\mathcal{Q}$ in certain implementations may be difficult due to issues such as misalignment errors. In such cases, even in the honest implementation, the statistics one ought to accept are fine-grained data that, because of complications, lack certain symmetries in Alice and Bob's results. This in turn limits one's a priori knowledge of what form the worst-case scenario observed statistics will take. This is further aggravated by the trade-off between how often the protocol will be aborted and the length of the secret key when the protocol does not abort. For these reasons, constructing a good choice of $\mathcal{Q}$ is a non-trivial task due to common issues in implementing QKD protocols. As it is designed for generic protocols, our numerical method allows for further exploration of these difficulties which cannot be explored analytically.
\end{subsection}

\end{section}

\begin{section}{Conclusion}\label{sec:conclusion}
In the utilization of QKD protocols for our future quantum-safe infrastructure, it is crucial that we can analyze general QKD protocols' ability to generate composable secret keys in the finite regime. Much work has already been done on both the framework of finite key analysis \cite{renner05,scarani08b,bratzik2011} as well as its analysis for specific protocols using both theory and numerics \cite{scarani2008,scarani09a,curty2014,lim2014,bunandar2019}. However, there has not existed a tool which can determine the finite key rate for any QKD protocol which can be represented in finite-dimensional Hilbert spaces. The contribution of this work has been to construct such a tool with the further properties that it always determines a secure secret key rate (reliability) and can in principle exactly determine the secret key rate under the security proof method presented in \cite{renner05}. 

We note that the tightness property of our solver is only up to the security proof method of \cite{renner05} where the smooth min-entropy is bounded by the conditional von Neumann entropy. However, it was shown in \cite{bratzik2011} that in some regimes bounding the smooth min-entropy by the min-entropy can improve the key rate. This method has also been implemented for a subclass of protocols numerically in \cite{bunandar2019}. Therefore, our claim of tightness is up to the assumption above, although it is easy to see one can unify our general framework with the min-entropy calculation presented in \cite{bunandar2019} and recover the tightness property up to this alternative choice. 

Furthermore, we note that it is not only easy to unify, but necessary for the application of the numerical method to general QKD protocols and obtain tightness within the proof method. This is the case because our proof of being able to consider multiple coarse-grainings at no cost in security parameter $\varepsilon_{\PE}$ and our introduction of the trace norm to handle multiple outcome POVMs without looseness is in some cases necessary to guarantee tight results. Furthermore, our method derives its practicality in implementations not only from its tightness, but from the ability to consider the acceptance set, $\mathcal{Q}$. As none of these tools are presented in \cite{bunandar2019}, it would lead to loose key rates for QKD protocols  with asymmetric observations as we saw in Section \ref{subsec:rotatedBB84} as well as not being applicable for practical implementations as it is designed only to consider protocols with unique acceptance. Therefore, this unifiication is necessary for general protocols.

Beyond the construction of a generic numerical framework for finite key analysis, we note that Theorem \ref{thm:multCoarseGrain} in this paper resolves an issue about this security framework for finite key analysis. If one were to define the security using only one set of statistics, as coarse-graining can be better than fine-grained data, it would follow that there exist cases in which Alice and Bob throwing out information is an advantage against Eve. This would be counter-intuitive. However, we see that the security definition actually would allow Alice and Bob to keep both versions of the data, and thus the `true' finite key rate can be seen as constraining over all possible coarse-grainings which utilizes all possible data from the experiment. The consideration of the rotated BB84 case exemplifies this idea.

Having presented a general numeric framework for finite key analysis which improves upon the pre-existing framework \cite{renner05}, we note two paths of research going forward. The first path is the application of this method to decoy state QKD protocols in the numerical framework. As previously discussed in \cite{coles2016}, for a discrete phase randomized source, or if one approximates continuous phase randomization by discrete phase randomization, one would simply consider a signal state for each intensity. In principle this could be immediately implemented following the numerical method used for the numerical analysis in Section \ref{subsec:DPRBB84}, but this will lead to large demands on the memory of the computer. Therefore, a better alternative approach for continuous phase randomization would be to consider `tagging' in which one fixes a photon number cutoff and treats multi-photon components above this cutoff as orthogonal states given to Eve. This block-diagonal structure can improve the cost on memory, but would require calculating the statistical fluctuations on the individual blocks.

The second path for future research follows from noting that this generic method requires that one considers probability distributions, but in CV-QKD one often is interested in a form of coarse-graining which leads to expectation values of observables rather than a probability distribution. One would hope there exists a proof method within the same security definitions which bounds the expectation values of these specific observables, even though they are not constructed using a conditional probability distribution applied to the initial fine-grained statistics. \\

\textit{Note added:} During the preparation of this manuscript, we noticed a similar work \cite{bunandar2019} is posted in the preprint server. Our ideas were conceived independently and we have presented many of our main results in a conference \cite{george19}. We point out that our work is different from Ref. \cite{bunandar2019} in that it considers an entry-wise bound on the trace norm for the variational bound and ignores the acceptance set $\mathcal{Q}$ altogether. This entry-wise bound introduces looseness when one considers fine-grained statistics and the latter limits it primarily to impractical QKD implementations.
\end{section}

\begin{acknowledgements}
I.G. would like to thank Jamie Sikora for fascinating discussions on semidefinite programming. The work has been performed at the Institute for Quantum Computing, University of Waterloo, which is supported by Innovation, Science and Economic Development Canada. The research has been supported by Natural Sciences and Engineering Research Council of Canada under the Discovery Grants Program, Grant No. 341495, and under the Collaborative Research and Development Program, Grant No. CRDP J 522308-17. Financial support for this work has been partially provided by Huawei Technologies Canada Co., Ltd.
\end{acknowledgements}

\onecolumngrid
\appendix

 \begin{section}{Numerical Methods Proofs}\label{appendix:NumMethods}
 In this Appendix we present the derivations and proofs for the finite key numerical method in detail.

\subsection{Notation}
 We begin with a brief explanation of notations used in this Appendix. 
For some arbitrary finite-dimensional Hilbert spaces $\mathcal{X}$ and $\mathcal{Y}$, $\Lin(\mathcal{X})$ denotes the set of linear maps from $\mathcal{X}$ to itself, $\Herm(\mathcal{X}) \subseteq \Lin(\mathcal{X})$ denotes the set of Hermitian operators acting on $\mathcal{X}$, $\Pos(\mathcal{X}) \subseteq \Herm(\mathcal{X})$ denotes the set of positive semidefinite operators, and $\mathrm{T}(\mathcal{X},\mathcal{Y})$ denotes the set of linear maps that map $\Lin(\mathcal{X})$ to $\Lin(\mathcal{Y})$. We use uppercase letters like $A$ and $B$ to denote matrices and lowercase letters like $z$ to denote complex (or real) numbers. For a vector $\vec{v}$, its $j$-th entry is denoted by $v(j)$. As already used in the main text, the inner product on $\Lin(\mathcal{X})$ is the Hilbert-Schmidt inner product, that is, $\langle A, B \rangle = \Tr(A^{\dagger}B)$ for $A, B \in \Lin(\mathcal{X})$. The norm $\norm{\cdot}_{\text{HS}}$ is the norm induced by the Hilbert-Schmidt inner product. For a Hermitian matrix $H$, let $\lambda_{\text{min}}(H)$ denote the minimum eigenvalue of $H$.

To ease the writing of matrices in block form, we introduce the following two shorthand notations: We write $\text{diag}(A_1, A_2)$ for the block-diagonal matrix $
\begin{pmatrix}
A_1 & 0 \\
0 & A_2
\end{pmatrix}$ where $A_1$ and $A_2$ are two square matrices (with possibly different sizes);  we write $\widetilde{\text{diag}}(A_1, A_2)$ for the matrix $\begin{pmatrix}
A_1 & B_1 \\
B_2 & A_2
\end{pmatrix}$ whose off-diagonal blocks are irrelevant for our discussion, where $B_1$ and $B_2$ are some arbitrary matrices of appropriate sizes. These two notations are generalized to a finite number of (at least two) square matrices.

For an arbitrary square matrix $A \in \Lin(\mathcal{X})$, $\text{diag}(A)$ denotes the vector whose entries are given by diagonal entries of $A$. For a vector $\vec{z}$, $\text{diag}(\vec{z})$ denotes the diagonal matrix whose diagonal entries are given by $\vec{z}$.

For any conditional probability distribution, $p_{\Lambda|\Sigma}$, which would be applied to a probability distribution $p_{\Sigma}$, there exists a completely-positive trace-preserving (CPTP) map representation $\mathcal{N}$ such that $\text{diag}(p_{\Lambda|\Sigma}p_{\Sigma}) = \mathcal{N}(\text{diag}(p_{\Sigma}))$ \cite{Wilde2011}. Explicitly, $\mathcal{N}(X) = \sum_{x \in \Sigma, y \in \Lambda} p(y|x)\ket{y}\bra{x}X\ket{x}\bra{y}$ and a straightforward calculation determines that the adjoint map is $\mathcal{N}^{\dagger}(Y) = \sum_{x \in \Sigma, y \in \Lambda} p(y|x) \bra{y}Y\ket{y}\dyad{x}{x}$. This will be useful in defining the SDP which involves processing probability distributions. For this reason in what follows we never define conditional probability distributions explicitly, but simply the corresponding CPTP map.

\subsection{Semidefinite Program Background}

 We give a short review the standard form of a semidefinite program and related concepts that will be useful in our proofs. \\

\textit{Definition:} \cite{watrous2018} Let $\Psi \in \mathrm{T}(\mathcal{X},\mathcal{Y})$ be a Hermitian-preserving map, $A \in \Herm(\mathcal{X})$, and $B \in \Herm(\mathcal{Y})$. A semidefinite program is a triple $(\Psi,A,B)$, with the following associated optimization problems:
\begin{equation}\label{eq:primaldef}
        \begin{aligned}
                & {\text{minimize}} & & \langle A, X \rangle \\
                & \text{subject to} & & \Psi(X) = B \\
                &                   & &  X \in \Pos(\mathcal{X})
        \end{aligned}
\end{equation}
\begin{equation}\label{eq:dualdef}
        \begin{aligned}
                & {\text{maximize}} & & \langle B, Y \rangle \\
                & \text{subject to} & & \Psi^{\dagger}(Y) \preceq A  \\
                &                   & &  Y \in \Herm(\mathcal{Y})
        \end{aligned}
\end{equation}
where $\Psi^{\dagger}$ is the adjoint map of $\Psi$; that is, $\Psi^{\dagger}$ is the unique linear map that satisfies the adjoint equation $\langle Y, \Psi(X) \rangle = \langle \Psi^{\dagger}(Y), X \rangle$ for every $X \in \Lin(\mathcal{X})$ and $Y \in \Lin(\mathcal{Y})$. Eqn. (\ref{eq:primaldef}) is referred to as the \textit{primal problem} and Eqn. (\ref{eq:dualdef}) is referred to as the \textit{dual problem}. \\

We define $\mathcal{A} = \{X \in \Pos(\mathcal{X}) | \Psi(X) = B\}$ and $\mathcal{B} = \{Y \in \Herm(\mathcal{Y}) | \Psi^{\dagger}(Y) \preceq A\}$. These sets are referred to as the \textit{feasible set} of the primal problem and dual problem, respectively. By \textit{weak duality}, for all semidefinite programs, the optimal value of the primal problem, denoted by $\alpha$, is always greater than or equal to the optimal value to the dual problem, denoted by $\beta$. If a semidefinite program has that $\alpha = \beta$, it is said to have \textit{strong duality}. A sufficient condition to show strong duality for SDP is Slater's condition for the standard form presented here.
\begin{theorem}{\textit{(Slater's Condition)}}
For a semidefinite program $(\Psi,A,B)$, if $\mathcal{A} \neq \emptyset$ and there exists a Hermitian operator $Y$ which \textit{strictly} satisfies the dual problem, that is, $\Psi^{\dagger}(Y) \prec A$, then $\alpha = \beta$ and the optimal value is obtained in the primal problem.
\end{theorem}

 \subsection{Numerical Imprecision}\label{sec:numerical_imprecision}
We recall two sets of constraints defined in the main text.  The set of constraints that are not subject to statistical fluctuation is denoted by $\{\Gamma_i | i \in \Lambda\}$ and we refer to these constraints as certainty constraints.  Constraints $\{\widetilde{\Gamma}_j | j \in \Sigma\}$ that are subject to statistical fluctuation are referred to as uncertainty constraints. 

As noted in Sec. \ref{sec:numMethods}, when one acquires a solution $\rho_f$ after the first step in Algorithm \ref{alg:Method}, the answer may not truly be feasible; that is, $\rho_{f}$ is not in the correct set $\mathbf{S}_{\mu}$, but rather in an enlarged set $\widetilde{\mathbf{S}}_{\mu}$. This issue arises from the imprecise numerical representation of the POVMs as well as the imprecision of the numerical optimization solver which lead to violation of constraints in the optimization problem. To resolve this issue, one needs to consider the larger set $\widetilde{\mathbf{S}}_{\mu}$ to guarantee that $\rho_{f}$ is included. Reference \cite{winick2018} presents a method for the asymptotic case. In Ref. \cite{winick2018}, one has to consider only violations pertaining to certainty constraints $\{\Gamma_i\}$. In the finite key scenario, we also need to consider the uncertainty constraints $\{\widetilde{\Gamma}_j\}$. To rigorously account for numerical imprecision, we now adapt the method in \cite{winick2018} to finite key analysis. 

An imprecise solver may lead to a solution $\rho_{f}$ which is not positive semidefinite or that does not satisfy these constraints. To handle the first issue, if the state $\rho_f$ has negative eigenvalues, one first perturbs the state to be $\rho_{f}' \equiv \rho_{f} + | \lambda_{\min}(\rho_{f}) | \bbone$ so that $\rho_f'$ does not have negative eigenvalues. Then one checks the maximum violation of the certainty constraints of $ \rho_{f}'$, and define $\epsilon_{\text{sol}}  \equiv \underset{i \in \Lambda}{\max} |\Tr(\rho_{f}' \Gamma_{i}) - \gamma_{i}| $. 

Imprecise representations can be seen as deviations from the true POVM and probability representations. One can therefore denote the imprecise representations as follows:
\begin{align*}
    \overline{\Gamma}_{i} = \Gamma_{i} + \delta \Gamma_{i} \text{ and } \overline{\gamma}_{i} = \gamma_{i} + \delta \gamma_{i},
\end{align*}
where $\norm{\delta \Gamma_i}_{\text{HS}} \leq \epsilon_{1}$ and $\abs{\delta \gamma_i } \leq \epsilon_{2}$ for all $i \in \Lambda$. By defining $\epsilon_{\text{rep}} \equiv \epsilon_{1} + \epsilon_{2}$, it is shown in Lemma 10 of Ref. \cite{winick2018} that $ | \Tr(\overline{\Gamma}_{i}) - \overline{\gamma}_{i} | \leq \epsilon_{\text{rep}}, \forall i \in \Lambda$. One then defines $\epsilon' = \max(\epsilon_{\text{sol}}, \ \epsilon_{\text{rep}})$ and considers $\rho$ subject to the constraints $\{| \Tr(\rho\overline{\Gamma}_i)-\overline{\gamma}_i | \leq \epsilon'\}$.

These imprecisions may also lead to violation of the variational distance constraint. Therefore, one should redefine $\mu$ for the second step to guarantee the $\rho_{f}$ is considered in the second step. Since the uncertainty constraints pertain to the variational distance which takes the imprecisions as a whole, to properly enlarge $\mu$ to take constraint violations into account, one can use the Cauchy-Schwarz inequality along with Lemma 10 of \cite{winick2018} to expand $\mu$ as $\mu' = \max(\mu + n  \epsilon', \| \Phi_{\mathcal{P}}(\rho_{f}) - F \|_{1} + n  \epsilon')$ where $n = |\Lambda|$. 

Lastly, there is the possibility that the solver finds an optimal solution $(\sigma,F)$ such that $\|\overline{\mathcal{N}}(F)-\overline{\mathcal{N}}(\overline{F})\|_{1} > t$. In this case, one should expand $t$. Thus define $t' \equiv \max(t,\|\overline{\mathcal{N}}(F)-\overline{\mathcal{N}}(\overline{F})\|_{1})$. Then one defines $ \mathbf{S}_{\mu' \epsilon' t'} $ to play the role of $\mathbf{S}_{\mu}$ by the following:
\begin{equation}\label{eq:feasibleSetwithImprecision}
    \mathbf{S}_{\mu'\epsilon' t'} = \{ \rho \in \Pos(\mathcal{H}_A \otimes \mathcal{H}_B) \hspace{0.1cm} | \hspace{0.1cm} |\Tr(\overline{\Gamma}_{i}\rho) - \overline{\gamma}_{i} | \leq \epsilon'  \hspace{0.1cm} \forall i \in \Lambda, \|\Phi_{\mathcal{P}}(\rho) - \mathcal{N}(F) \|_{1} \leq \mu', \|\overline{\mathcal{N}}(F) - \overline{\mathcal{N}}(\overline{F}) \|_{1} \leq t' \} \supseteq \mathbf{S}_{\mu}
\end{equation}
Clearly, if $\epsilon' = 0$, $t' = t$, and $\mu' = \mu$, one reconstructs the original set $\mathbf{S}_{\mu}$. This alternative set is used for deriving the dual problem in the second step in the following section. By optimizing over this set $\mathbf{S}_{\mu' \epsilon' t'} $, we handle the numerical imprecision related to certainty and uncertainty constraints.

A final remark is that when $\mathcal{G}(\rho)$ is singular, the derivative in Eqn. (\ref{eq:derivative}) may not exist. To tackle this issue, Ref. \cite{winick2018} introduces a small perturbation as 
\begin{equation}\label{eq:perturbedObjective}
\begin{aligned}
\mathcal{G}_{\epsilon}(\rho) & \equiv (1-\epsilon) \mathcal{G}(\rho) + \epsilon \bbone/d', \\
f_{\epsilon}(\rho) &\equiv D\big(\mathcal{G}_{\epsilon}(\rho)|| \mathcal{Z}[\mathcal{G}_{\epsilon}(\rho)]\big),
\end{aligned}
\end{equation}where $d'$ is the dimension of $\mathcal{G}(\rho)$, and $\epsilon \geq 0$ is chosen in a way such that $\mathcal{G}_{\epsilon}(\rho)$ is not singular. The derivative of $f_{\epsilon}(\rho) $ is obtained by replacing $\mathcal{G}$ with $\mathcal{G}_{\epsilon}$ in Eqn. (\ref{eq:derivative}).

 \subsection{Finite Key SDP}
 
We present the SDP that also takes into account the numerical imprecision discussed above. (However, for ease of writing, we still use $\{\Gamma_i\}$ to denote certainty constraints and $\{\widetilde{\Gamma}_j\}$ to denote uncertainty constraints.) For simplicity, we present here derivations in the case of one variation bound and state the result related to multiple coarse-grainings in Sec. \ref{sec:multiple_POVMs}. 
 
The primal problem of our SDP at $\rho \in \mathbf{S}_{\mu'\epsilon' t'}$ is
\begin{equation}\label{eq:rigorousFiniteSDPPrimal}
        \begin{aligned}
                & {\text{minimize}} & & \langle \nabla f_\epsilon(\rho), \sigma \rangle & \\
                & \text{subject to}                     & & \Tr(G) + \Tr(H) \leq \mu' \\
                &                   & & G  \succeq  \Phi_{\mathcal{P}}(\sigma) - \mathcal{N}(F)  \\
                &                   & & H \succeq  \mathcal{N}(F) - \Phi_{\mathcal{P}}(\sigma)   \\
                &                   & & \Tr(\overline{G}) + \Tr(\overline{H}) \leq t' \\
                &                   & & \overline{G}  \succeq  \overline{\mathcal{N}}(F) - \overline{F}_{\mathcal{\overline{N}}}  \\
                &                   & & \overline{H} \succeq  \overline{F}_{\mathcal{\overline{N}}}  - \overline{\mathcal{N}}(F)    \\
                &                   & & \Tr(F) = 1 \\
                &                   & & \abs{\Tr(\Gamma_{i} \sigma) - \gamma_i} \leq \epsilon' \; \forall i \in \Lambda \\
                &                   & & \sigma,F, G, H, \overline{G},\overline{H} \succeq 0.
        \end{aligned}
\end{equation} where $\overline{F}_{\mathcal{\overline{N}}} \equiv \overline{\mathcal{N}}(\overline{F})$. We use this notation to emphasize $\overline{\mathcal{N}}(\overline{F})$ is fixed and is not an optimization variable because $\overline{F}$ and $\overline{\mathcal{N}}$ are both fixed. We note this is Eqn. \ref{eq:SDfiniteSDPPrimal} with the inclusion of numerical imprecision. This equation therefore considers the set of density matrices which define collective attacks Alice and Bob would non-negligibly accept (see Section \ref{subsec:BackgroundPE} for further discussion), but with the numerical imprecision of the computer taken into account. Let $\alpha_0(\rho)$ denote the optimal value of this primal problem. To derive its dual problem, Eqn. (\ref{eq:rigorousFiniteSDPPrimal}) can be reformatted to fit the definition of Eqn. (\ref{eq:primaldef}) as follows:
\begin{equation}\label{eq:MapDefn}
\begin{aligned}
    &A = \text{diag}(\nabla f_{\epsilon}(\rho),\overline{\mathbf{0}}) \\
    &B= \text{diag}(\mu', 0, 0, t', \overline{F}_{\mathcal{\overline{N}}}, - \overline{F}_{\mathcal{\overline{N}}}, 1, \sum_{i}(\epsilon + \gamma_i)\ket{i}\bra{i}, \sum_i (\epsilon-\gamma_i)\ket{i}\bra{i}) \\
    &\Psi(X)=  \text{diag}(\Tr(G) + \Tr(H) + z, -G - \mathcal{N}(F) + \Phi_{\mathcal{P}}(\sigma) + I, -H + \mathcal{N}(F) - \Phi_{\mathcal{P}}(\sigma) + J, \Tr(\overline{G}) + \Tr(\overline{H}) + \overline{z}, \\ & \hspace{3cm}  -\overline{G} + \overline{\mathcal{N}}(F) + \overline{I}, -\overline{H} - \overline{\mathcal{N}}(F) + \overline{J}, \Tr(F), \Phi_{0}(\sigma) + M_{1}, -\Phi_{0}(\sigma) + M_{2}) \\
    &X = \widetilde{\text{diag}}(\sigma, F, G, H, z, I, J, \overline{G}, \overline{H}, \overline{z}, \overline{I}, \overline{J}, M_{1}, M_{2})
\end{aligned}
\end{equation}
where $\overline{\mathbf{0}}$ is a shorthand notation to mean that all other blocks are zero matrices of appropriate size, $\Phi_{0}(X) \equiv \sum_{i \in \Lambda} \Tr(X \Gamma_{i})\dyad{i}{i}$, $\mathcal{N}(X) = \sum_{x,y} p(y|x)\ket{y}\bra{x} X \ket{x}\bra{y}$, $\overline{\mathcal{N}}(X) = \sum_{x,y} \overline{p}(y|x)\ket{y}\bra{x} X \ket{x}\bra{y}$,  and $z,\overline{z} \in \mathbb{C}, I \in \Lin(\mathbb{C}^{\abs{\Sigma}}), J \in \Lin(\mathbb{C}^{\abs{\Sigma}})$, $\overline{I} \in \Lin(\mathbb{C}^{\abs{\Sigma_{C}}})$, $\overline{J} \in \Lin(\mathbb{C}^{\abs{\Sigma_{C}}})$,  $M_1 \in \Lin(\mathbb{C}^{\abs{\Lambda}})$ and $M_2 \in \Lin(\mathbb{C}^{\abs{\Lambda}})$ are slack variables. Furthermore $\Sigma_{C}$ represents the alphabet for the coarse-graining. It is easy to verify using the definition of adjoint map, $\langle Y, \Psi(X) \rangle = \langle \Psi^{\dagger}(Y), X \rangle$, that the adjoint of $\Psi$ is:
\begin{equation}
\begin{aligned}
     \Psi^{\dagger}(Y) &= \text{diag}(\Phi_{0}^{\dagger}(W_{1} - W_{2}) + \Phi_{\mathcal{P}}^{\dagger}(K - L), \mathcal{N}^{\dagger}(L - K) + \overline{\mathcal{N}}^{\dagger}(\overline{K} - \overline{L}) + b \bbone_{\mathcal{W}}, a \bbone_{\mathcal{W}} - K, \\
     & \hspace{1.5cm} a\bbone_{\mathcal{W}} - L, a, K,  L, \overline{a} \bbone_{\mathcal{W}} - \overline{K}, \overline{a} \bbone_{\mathcal{W}} - \overline{L}, \overline{a}, \overline{K}, \overline{L} ,W_1 , W_2)
\end{aligned}
\end{equation}
where $Y = \widetilde{\text{diag}}(a, K, L, \overline{a}, \overline{K}, \overline{L}, b, W_1, W_2)$,
\begin{equation}
\Phi^{\dagger}_0 (W) = \sum_{i \in \Lambda} W(i, i) \Gamma_i, \hspace{0.3cm}  \Phi^{\dagger}_{\mathcal{P}} (V) = \sum_{j \in \Sigma} V(j, j) \widetilde{\Gamma}_j
\end{equation}

If we substitute these definitions in the standard form of SDP [in Eqns. (\ref{eq:primaldef}) and (\ref{eq:dualdef})] and flip signs of $a, \overline{a}, b, K, L, \overline{K}$, and $\overline{L}$, we then get the following dual problem:
\begin{equation}\label{eq:matrixDualSDP}
        \begin{aligned}
                & {\text{maximize}} & & \langle \sum_{i \in \Lambda} (\epsilon' + \gamma_i) \dyad{i}{i}, W_1 \rangle +  \langle \sum_{i \in \Lambda} (\epsilon' - \gamma_i)\dyad{i}{i}, W_2 \rangle  + \langle \overline{F}_{\mathcal{\overline{N}}}, \overline{L}-\overline{K} \rangle- \mu' a - t'\overline{a} - b\\
                & \text{subject to} & & \sum_{i \in \Lambda} [W_1(i,i)-W_2(i,i)]\Gamma_i + \sum_{j \in \Sigma} [L(j,j)-K(j,j)] \widetilde{\Gamma}_j \preceq \nabla f_{\epsilon}(\rho)\\
                &                   & & \overline{\mathcal{N}}^{\dagger}(\overline{L} - \overline{K}) - \mathcal{N}^{\dagger}(L - K) \preceq b \bbone_{\mathcal{W}} \\
                &                   & & 0 \preceq K \preceq a \bbone_{\mathcal{W}} \hspace{2cm} 0 \preceq \overline{K} \preceq \overline{a} \bbone_{\mathcal{W}}\\
                &                   & & 0 \preceq L \preceq a \bbone_{\mathcal{W}} \hspace{2cm} 0 \preceq \overline{L} \preceq \overline{a} \bbone_{\mathcal{W}} \\
                &                   & & a,\overline{a} \geq 0, \hspace{1cm} W_1, W_2 \preceq 0,
        \end{aligned}
\end{equation}where $\mathcal{W} \equiv \mathbb{C}^{\abs{\Sigma}}$. Let $\beta_0(\rho)$ denote the optimal value of this dual problem.

From Eqn. (\ref{eq:matrixDualSDP}), we observe that off-diagonal entries of $K$, $L$, $\overline{K}$, $\overline{L}$, $W_1$ and $W_2$, are not important for this optimization problem since for any optimal solution $Y^*=\widetilde{\text{diag}}(a^*, K^*, L^*,\overline{a}^{*},\overline{K}^{*},\overline{L}^{*},b^{*}, W_1^*, W_2^*)$ of this problem, if  $K'$, $L'$,$\overline{K}'$,$\overline{L}'$, $W_1'$ and $W_2'$ are matrices obtained by taking only the diagonal parts of $K^*, L^*,\overline{K}^{*},\overline{L}^{*}, W_1^*$ and $W_2^*$, respectively, then the matrix $Y' = \text{diag}(a^*, K', L',\overline{K}',\overline{L}', W_1',W_2')$ is also optimal as it is feasible and achieves the same optimal value. Moreover, we may optimize over the difference $L-K$ ($\overline{L}-\overline{K}$) subject to the constraint $-a \bbone_{\mathcal{W}} \preceq L-K \preceq a \bbone_{\mathcal{W}}$ ($-\overline{a} \bbone_{\mathcal{W}} \preceq \overline{L}-\overline{K} \preceq \overline{a} \bbone_{\mathcal{W}}$) as only the difference $L-K$ ($\overline{L}-\overline{K}$) matters in the optimization and its range is $-a\bbone \preceq L-K \preceq a\bbone$ ($-\overline{a}\bbone \preceq \overline{L} - \overline{K} \preceq \overline{a}\bbone$) which is determined by the two constraints $0 \preceq K \preceq a\bbone$ and $0 \preceq L \preceq a \bbone$ ($0 \preceq \overline{K} \preceq \overline{a}\bbone$ and $0 \preceq \overline{L} \preceq \overline{a} \bbone$). If we write $\vec{\gamma}$ as the vector whose $i$-th entry is $\gamma_i$ and $\overline{f} = \text{diag}(\overline{F}_{\mathcal{\overline{N}}})$, the dual problem in Eqn. (\ref{eq:matrixDualSDP}) is simplified as
\begin{equation}\label{eq:vecSDPFinal}
        \begin{aligned}
                & {\text{maximize}} & & (\epsilon' + \vec{\gamma}) \cdot \vec{y_1} + (\epsilon' - \vec{\gamma}) \cdot \vec{y_2}+ \overline{f} \cdot \vec{\overline{z}} -\mu'  a - t' \overline{a} - b \\
                & \text{subject to} & & \sum_{i \in \Lambda}  [y_1(i) - y_2(i)]\Gamma_i + \sum_{j \in \Sigma} z(j) \widetilde{\Gamma}_j \preceq \nabla f_{\epsilon}(\rho)\\
                &                   & &               \overrightarrow{\overline{N}^{\dagger}}(\vec{\overline{z}}) - \overrightarrow{N^{\dagger}}(\vec{z}) \preceq b \vec{1} \\
                &                   & & -a \vec{1} \leq \vec{z} \leq a \vec{1}\\
                &                   & & -\overline{a} \vec{1} \leq \vec{\overline{z}} \leq \overline{a} \vec{1}\\
                &                   & &a,\overline{a} \geq 0 \hspace{1cm} \vec{y_1} , \vec{y_2} \leq \vec{0}. \\
        \end{aligned}
\end{equation}
where $\overrightarrow{\mathcal{N}^{\dagger}}$ is defined such that $\text{diag}(\mathcal{N}^{\dagger}(Z)) = \overrightarrow{\mathcal{N}^{\dagger}}(\text{diag}(Z))$ for arbitrary $Z \in \Lin(\mathbb{C}^{|\Sigma_{C}|})$. We remark that when $\epsilon'= 0$, we can replace $\vec{y_1}$ and $\vec{y_2}$ by $\vec{y}\equiv \vec{y_1} - \vec{y_2}$ subject to the constraint $\vec{y} \in \mathbb{R}^{|\Lambda|}$. When $\mu' = \mu$, $t' = t$, and $\epsilon' = 0$, Eqn. (\ref{eq:vecSDPFinal}) reduces to Eqn. (\ref{eq:finiteSDPDual}) in the main text after this replacement. 

\subsection{Reliability and Tightness}
We now prove that the lower bound using the linearization is tight for the finite key SDP. That is, in the limit where the numerical imprecisions go away, the program will obtain the true answer. In this section we present the precise mathematical statement of tightness for the SDP in Eqn. \ref{eq:finiteSDPPrimal} in Theorem \ref{thm:finalTightness} which considers the issues of numerical imprecision discussed in Sec. \ref{sec:numerical_imprecision}. The extension to multiple coarse-grainings is then straightforward. This theorem is a finite-size version of Theorem 3 in Ref. \cite{winick2018}. In proving this theorem, we will adapt the proofs in Appendixes D and E of \cite{winick2018} as well as technical lemmas in Appendixes A-C of \cite{winick2018}.

 As our optimization problem comes from a physical scenario and we are only interested in the situation where the set $\mathbf{S}_{\mu'\epsilon' t'}$ is not empty (otherwise we may trivially set the key rate to be zero), we restrict our attention to this situation. 
\begin{theorem}\label{thm:finalTightness} (General Proof of Tightness of Numerical Method) Let $\mathbf{S}_{\mu'\epsilon' t'}$ be defined in Eqn. (\ref{eq:feasibleSetwithImprecision}) and assume  $\mathbf{S}_{\mu'\epsilon' t'} \neq \emptyset$. Let $\rho \in \mathbf{S}_{\mu'\epsilon' t'}$ where $\mathcal{G}(\rho)$ is of size $d' \times d'$ and $\epsilon' > 0$. For $0 < \epsilon \leq 1/[e(d'-1)]$, then
\begin{equation}\label{eq:reliability}
\alpha \geq \beta_{\mu' \epsilon' t' \epsilon}(\rho) - \zeta_{\epsilon}
\end{equation}where 
\begin{equation}
\alpha = \min_{\sigma \in \mathbf{S}_{\mu}} f(\sigma),
\end{equation}
\begin{equation}
\beta_{\mu' \epsilon' t' \epsilon}(\sigma) \equiv f_{\epsilon}(\sigma) - \Tr[\sigma \nabla f_{\epsilon}(\sigma)] + \underset{(a,\overline{a}, \vec{y_1}, \vec{y_2},\vec{z},\vec{\overline{z}},b) \in \mathbf{S}^*_{\mu' \epsilon' t'}(\sigma)}{\max} [(\epsilon' + \vec{\gamma}) \cdot \vec{y_1} + (\epsilon' - \vec{\gamma}) \cdot \vec{y_2}+ \overline{f} \cdot \vec{\overline{z}} -\mu'  a - t' \overline{a} - b],
\end{equation} and
\begin{equation}
\zeta_{\epsilon} \equiv 2\epsilon (d'-1) \log_2\frac{d'}{\epsilon(d'-1)}.
\end{equation}
The set $\mathbf{S}^*_{\mu' \epsilon' t'}(\sigma)$ is defined by
\begin{equation}
\begin{aligned}
\mathbf{S}^*_{\mu' \epsilon' t'}(\sigma)
\equiv&  \{(a, \overline{a}, \vec{y_1},\vec{y_2},\vec{z},\vec{\overline{z}},b) \in (\mathbb{R},\mathbb{R}, \mathbb{R}^{|\Lambda|}, \mathbb{R}^{|\Lambda|},\mathbb{R}^{|\Sigma |},\mathbb{R}^{|\Sigma |},\mathbb{R}) | \\  & \hspace{0.5cm} a,\overline{a} \geq 0, -a \vec{1} \leq \vec{z} \leq a\vec{1},-\overline{a} \vec{1} \leq \vec{\overline{z}} \leq \overline{a}\vec{1}, \vec{y_1} \leq 0, \vec{y_2} \leq 0, \\ & \hspace{0.75cm} \sum_{i \in \Lambda}  [y_1(i) - y_2(i)]\Gamma_i + \sum_{j \in \Sigma} z(j) \widetilde{\Gamma}_j \preceq \nabla f_{\epsilon}(\sigma), \overrightarrow{\overline{N}^{\dagger}}(\vec{\overline{z}}) - \overrightarrow{N^{\dagger}}(\vec{z}) \preceq b \vec{1} \}
\end{aligned}
\end{equation}
Moreover, if $\rho^*$ is an optimal solution to the primal problem,
\begin{equation}\label{eq:tightness}
\lim_{\epsilon\rightarrow 0+} \lim_{ \substack{\epsilon' \rightarrow 0+\\ \mu' \rightarrow \mu \\ t' \rightarrow t}} [\beta_{\mu' t'\epsilon\epsilon'}(\rho^*) -\zeta_{\epsilon}] = \alpha.
\end{equation}
\end{theorem}
We note that the statement of tightness in the main text (Theorem \ref{thm:MainTextTightness}) is for when there are no numerical imprecisions. Theorem \ref{thm:finalTightness} is a generalization of that theorem that handles numerical imprecisions as well. \\

To prove Theorem \ref{thm:finalTightness}, we first show that for any $\rho \in \mathbf{S}_{\mu' \epsilon' t'}$, the primal optimal value $\alpha_0(\rho)$ is equal to the dual optimal value $\beta_0(\rho)$ as Lemma \ref{lemma:strongduality}. Then, we break down the proof of theorem into two parts: reliability in Eqn. (\ref{eq:reliability}) and tightness in Eqn. (\ref{eq:tightness}).

\begin{lemma}\label{lemma:strongduality}
If $\mathbf{S}_{\mu' \epsilon' t'} \neq \emptyset$, then $\alpha_0(\rho) = \beta_0(\rho)$ for any $\rho \in \mathbf{S}_{\mu' \epsilon' t'}$.
\end{lemma}
\begin{proof}As $\mathbf{S}_{\mu'\epsilon' t'} \neq \emptyset$, to apply Slater's condition, we just find a strictly feasible solution to the dual problem. We consider the dual problem in the form of Eqn. (\ref{eq:matrixDualSDP}). Let $a = \overline{a} = 3$. Let $W_{1} = \text{diag}(x-3,-1,-1,...,-1)$ where $x = - |\lambda_{\min}(\nabla f_{\epsilon}(\rho))|$. Thus $W_{1} \leq 0$. Let $W_{2} = -\bbone \leq 0$. Without loss of generality, let $\Gamma_{1} = \bbone$ as we always have the constraint $\Tr(\sigma) = 1$ in the primal problem. Let $L = 2\bbone_{\mathcal{W}}$ and $K=\bbone_{\mathcal{W}}$.Thus $-a \bbone_{\mathcal{W}} \prec L - K \prec a \bbone_{\mathcal{W}}$. Furthermore, $\sum_{j} [K(j,j) - L(j,j)]\widetilde{\Gamma}_{j} = \bbone$ as $\{\widetilde{\Gamma}_{j}\}$ is a POVM. Thus, $\sum_{i}[W_{1}(i,i) - W_{2}(i,i)]\Gamma_{i} + \sum_{j} [K(j,j) - L(j,j)]\widetilde{\Gamma}_{j} = (x-1) \bbone \prec \nabla f_{\epsilon}(\rho)$ by construction of $x$. Let $\overline{L} = 2\bbone_{\mathcal{W}}$, $\overline{K} = \bbone_{\mathcal{W}}$ and $b = 2$. Then $-\overline{a} \bbone_{\mathcal{W}} \prec \overline{L} - \overline{K} \prec \overline{a} \bbone_{\mathcal{W}}$ and $\overline{\mathcal{N}}^{\dagger}(\overline{L} - \overline{K}) - \mathcal{N}^{\dagger}(L - K) = 0 \prec b \bbone_{\mathcal{W}}$. The last equality followed from the fact $\mathcal{N}$ is a quantum channel and so its adjoint is unital. Thus all inequalities are strictly satisfied. 
\end{proof}

We now adapt the proof in Appendix D.3 of \cite{winick2018} to finite-key scenario.
\begin{lemma}\label{lemma:reliability}
In the context of Theorem \ref{thm:finalTightness},  $\alpha \geq \beta_{\mu' \epsilon' t' \epsilon}(\rho) - \zeta_{\epsilon}$ for any $\rho \in \mathbf{S}_{\mu'\epsilon' t'}$, which is Eqn. (\ref{eq:reliability}). 
\end{lemma}
\begin{proof}
Let $\alpha_{\mu'\epsilon' t' \epsilon} \equiv \underset{\sigma\in\mathbf{S}_{\mu'\epsilon' t'}}{\min}f_{\epsilon}(\sigma)$. Suppose that $\rho^*_{\mu'\epsilon' t' \epsilon} \in \mathbf{S}_{\mu'\epsilon' t'}$ is an optimal solution of this optimization. For any $\rho \in \mathbf{S}_{\mu'\epsilon' t'}$, since $f_{\epsilon}$ is convex, 
\begin{equation}\label{eq:perturbedWeakDualityLike}
\begin{aligned}
\alpha_{\mu'\epsilon' t' \epsilon}  = f_{\epsilon}(\rho^*_{\mu'\epsilon 't'\epsilon} ) &\geq f_{\epsilon}(\rho) + \langle (\rho^*_{\mu'\epsilon' t' \epsilon} - \rho), \nabla f_{\epsilon}(\rho) \rangle\\
& \geq f_{\epsilon}(\rho) - \langle  \rho, \nabla f_{\epsilon}(\rho) \rangle+ \min_{\sigma \in \mathbf{S}_{\mu'\epsilon' t'}} \langle \sigma, \nabla f_{\epsilon}(\rho) \rangle\\
& =  f_{\epsilon}(\rho)- \langle  \rho, \nabla f_{\epsilon}(\rho) \rangle+ \alpha_0(\rho)\\
& = f_{\epsilon}(\rho)- \langle  \rho, \nabla f_{\epsilon}(\rho) \rangle+ \beta_0(\rho) = \beta_{\mu' t'\epsilon' \epsilon}(\rho),
\end{aligned}
\end{equation}where first two inequalities follow from the same argument about this linearization of our convex objective function as it is used in Eqns. (77)-(79) of Ref. \cite{winick2018} and the last line follows from Lemma \ref{lemma:strongduality} and the definition of $\beta_{\mu' t'\epsilon' \epsilon}(\rho)$. Since $\mathbf{S}_{\mu} \subseteq \mathbf{S}_{\mu'\epsilon' t'}$,
\begin{equation}
\alpha = \min_{\sigma \in \mathbf{S}_{\mu}} f(\sigma) \geq \min_{\sigma \in \mathbf{S}_{\mu'\epsilon't'}} f(\sigma) \geq \min_{\sigma \in \mathbf{S}_{\mu'\epsilon' t'}} f_{\epsilon}(\sigma) - \zeta_{\epsilon} = \alpha_{\mu'\epsilon' t' \epsilon}- \zeta_{\epsilon},
\end{equation}where the last inequality follows from a continuity argument (which is Lemma 8 and Lemma 9 in Ref. \cite{winick2018}). Combining this result with Eqn. (\ref{eq:perturbedWeakDualityLike}) leads to Eqn. (\ref{eq:reliability}).
\end{proof}

As we have shown the reliability of our numerical method, we now proceed with the tightness in Eqn. (\ref{eq:tightness}). If $\rho^*$ is an optimal solution, an immediate consequence of Lemma \ref{lemma:reliability} is that for any $\rho \in \mathbf{S}_{\mu' \epsilon' t'}$, the following equation holds:
\begin{equation}\label{eq:lessthanzero}
  \underset{\sigma \in \mathbf{S}_{\mu'\epsilon' t'}}{\min} \Tr[(\sigma-\rho^*) \nabla f(\rho^*)] \leq 0.
  \end{equation}

As Eqn. (\ref{eq:lessthanzero}) holds for any feasible density operator in the set $\mathbf{S}_{\mu' \epsilon' t'} \supseteq \mathbf{S}_{\mu}$, we want to show that if $\rho^*$ optimizes the objective function $f$, then $ \underset{\sigma \in \mathbf{S}_{\mu}}{\min}  \Tr[(\sigma-\rho^*) \nabla f(\rho^*)] = 0$ where the optimization is over $\mathbf{S}_{\mu}$ as Eqn. \ref{eq:tightness} pertains to the limit where that is the set we are interested in. Therefore we just need to prove 
\begin{equation}\label{eq:morethanzero}
    \underset{\sigma \in \mathbf{S}_{\mu'\epsilon' t'}}{\min} \Tr[(\sigma - \rho^{*})\nabla f(\rho^{*})] \geq 0    
\end{equation} 
when $\mathbf{S}_{\mu'\epsilon' t'} \neq \emptyset$.
\begin{lemma}\label{lemma:greaterThanZero} When $\mathbf{S}_{\mu'\epsilon' t'} \neq \emptyset$, 
\begin{equation}
    \underset{\sigma \in \mathbf{S}_{\mu'\epsilon' t'}}{\min} \Tr[(\sigma - \rho^{*})\nabla f(\rho^{*})] \geq 0    
\end{equation} 
\end{lemma}
\begin{proof}
Let $\mathbf{S}_{\mu'\epsilon' t'} \neq \emptyset$. By Lemma \ref{lemma:strongduality}, we know that Eqn. \ref{eq:rigorousFiniteSDPPrimal} obtains its optimal value. Let $\rho^{*}$ optimize $f$ over $\mathbf{S}_{\mu'\epsilon' t'} \neq \emptyset$. As $f$ is a differentiable, convex function (one may consider $f_{\epsilon}$ to guarantee differentiability), it is the case that for all $\sigma \in \mathbf{S}_{\mu'\epsilon' t'}$, $\Tr[\nabla f_{\epsilon'}(\rho^{*})(\sigma - \rho^{*})] \geq 0$ (Eqn. 4.21 of \cite{boyd2004}). It follows $\underset{\sigma \in \mathbf{S}_{\mu'\epsilon' t'}}{\min} \Tr[(\sigma - \rho^{*})\nabla f(\rho^{*})] \geq 0$ 
\end{proof}

Eqn. (\ref{eq:lessthanzero}) and Lemma \ref{lemma:greaterThanZero} imply that, given $\rho^{*}$ that optimizes $f$ over $\mathbf{S}_{\mu' \epsilon' t'}$, 
$$ \underset{\sigma \in \mathbf{S}_{\mu' \epsilon' t'}}{\min} \Tr((\sigma-\rho^{*}) \nabla f(\rho^{*})) = 0$$ We can therefore conclude the following:
\begin{align*}
    f(\rho^{*}) &= f(\rho^{*}) + \underset{\sigma \in \mathbf{S}_{\mu' \epsilon' t'}}{\min} \Tr[(\sigma-\rho^{*}) \nabla f(\rho^{*})] \\
    &= f(\rho^{*}) - \Tr(\rho^{*}\nabla f(\rho^{*})) + \underset{(a, \overline{a}, \vec{y_1},\vec{y_2},\vec{z},\vec{\overline{z}},b) \in \mathbf{S}_{\mu' \epsilon' t'}^{*}(\rho^{*})}{\max} [(\epsilon' + \vec{\gamma}) \cdot \vec{y_1} + (\epsilon' - \vec{\gamma}) \cdot \vec{y_2}+ \overline{f} \cdot \vec{\overline{z}} -\mu'  a - t' \overline{a} - b] \\
    &= \beta(\rho^{*})
\end{align*}
this completes the proof of Eqn. \ref{eq:tightness} and Theorem \ref{thm:finalTightness}.

\subsection{Multiple Coarse-Grainings}\label{sec:multiple_POVMs}
We now can show that it is easy to extend to the case where one considers multiple coarse-grainings. First, we define $\Sigma_{f}$ as the alphabet indexing the fine-grained statistics of the experiment. Let $k$ index the set of conditional probability distributions pertaining to coarse-grained data, $\{p_{\Sigma_{k}|\Sigma_{f}}\}_{k}$. Each conditional probability distribution induces a channel $\mathcal{N}_{k}$ which applies the coarse-graining to the statistics. Define the POVM which pertains to the $k^{th}$ conditional probability distribution as $\{\widetilde{\Gamma}^{k}_{j}\}_{j \in \Sigma_{k}}$ which induces a measurement channel $\Phi_{\mathcal{P}_{k}}$. In this case $j$ is implicitly dependent on $k$ as different coarse-grainings will construct probability distributions of different sizes. Then, the primal problem may be written as:
\begin{equation}\label{eq:finiteSDPPrimalMulti}
        \begin{aligned}
                & {\text{minimize}} & & \langle \nabla f_{\epsilon}(\rho), \sigma \rangle & \\
                & \text{subject to} & & \Tr(\Gamma_{i} \sigma) = \gamma_{i} & \forall i \in \Lambda \\
                &                   & & \|  \Phi_{\mathcal{P}_k}(\sigma) - \mathcal{N}_{k}(F_{k}) \|_{1} \leq \mu_{k}  & \forall k  \\
                &                   & & \| \overline{\mathcal{N}}(F_{k}) - \overline{\mathcal{N}}(\overline{F})\|_{1} \leq t & \forall k\\
                &                   & & F_{k} \succeq 0 & \forall k \\
                &                   & & \sigma \succeq 0
        \end{aligned}
\end{equation}
where $\Phi_{\mathcal{P}_k}(X) \equiv \sum_{j \in \Sigma_{k}} \Tr(X \widetilde{\Gamma}_{j}^{k})\ket{j}\bra{j}$, $\mathcal{N}_{k}(X) = \sum_{x \in \Sigma_{f}, y \in \Sigma_{k}} p_{\Sigma_{k}|\Sigma_{f}}(y|x)\ket{y}\bra{x}X\ket{x}\bra{y}$. We stress that $F_{k}$ is indexed by $k$ given the set considered in Theorem \ref{thm:multCoarseGrain}. \\

To convert this linearized primal problem into a semidefinite program, we effectively are just optimizing $k$ copies of Eqn. \ref{eq:rigorousFiniteSDPPrimal} at the same time. This means we can write the equivalent form of Eqn. (\ref{eq:rigorousFiniteSDPPrimal}):
\begin{equation}\label{eq:rigorousFiniteSDPPrimalMulti}
        \begin{aligned}
                & {\text{minimize}} & & \langle \nabla f_\epsilon(\rho), \sigma \rangle & \\
                & \text{subject to}         & & \Tr(G_{k}) + \Tr(H_{k}) \leq \mu'_{k} \hspace{0.25cm} \forall k \\
                &                   & & G_{k}  \succeq \Phi_{\mathcal{P}_k}(\sigma)- F_k \hspace{0.25cm} \forall k \\
                &                   & & H_{k} \succeq  F_k - \Phi_{\mathcal{P}_k}(\sigma) \hspace{0.25cm} \forall k \\
                &     & & \Tr(\overline{G}_{k}) + \Tr(\overline{H}_{k}) \leq t_{k}' \hspace{0.25cm} \forall k \\
                &                   & & \overline{G}_{k}  \succeq \overline{\mathcal{N}}(F_{k})- \overline{F}_{\mathcal{\overline{N}}} \hspace{0.25cm} \forall k \\
                &                   & & \overline{H}_{k} \succeq  \overline{F}_{\mathcal{\overline{N}}} - \overline{\mathcal{N}}(F_{k}) \hspace{0.25cm} \forall k \\
                &                   & & |\Tr(\Gamma_{i} \sigma) - \gamma_i | \leq \epsilon' \\
                &                   & & F_{k}, G_{k}, H_{k}, \overline{G}_{k}, \overline{H}_{k} \succeq 0 \hspace{0.25cm} \forall k \\
                &                   & & \sigma \succeq 0
        \end{aligned}
\end{equation}
where we have let $t'_{k}$ be indexed by $k$ in case different coarse-grainings violate the $\mathcal{Q}$ set by different amounts.

To reformat Eqn. (\ref{eq:rigorousFiniteSDPPrimalMulti}) into the definition in Eqn. (\ref{eq:primaldef}) we can extend the definitions in Eqn. (\ref{eq:MapDefn}) in a block diagonal fashion using the matrix direct sum, $\oplus$, over $k$. 
\begin{align*}
    A &= \text{diag}(\nabla f_{\epsilon}(\rho), \overline{0}) \\
    B &= \text{diag}(\oplus_{k} \mu_{k}, \oplus_{k} 0,\oplus_{k} 0, \oplus_{k} t, \oplus_{k} \overline{F}_{\mathcal{\overline{N}}}, \oplus_{k} - \overline{F}_{\mathcal{\overline{N}}},\oplus_{k} 1, \sum_{i}(\epsilon + \gamma_i)\ket{i}\bra{i}, \sum_i (\epsilon-\gamma_i)\ket{i}\bra{i}) \\
    \Psi(X) &=  \text{diag}(\oplus_k[\Tr(G_{k}) + \Tr(H_k) + z_k], \oplus_k[-G_{k} - \mathcal{N}_{k}(F_{k}) + \Phi_{\mathcal{P}_{k}}(\sigma) + I_k], \oplus_k [-H_{k} + \mathcal{N}_{k}(F_{k}) - \Phi_{\mathcal{P}_{k}}(\sigma) + J_k], \\ & \hspace{2cm} \oplus_k[\Tr(\overline{G}_{k}) + \Tr(\overline{H}_k) + \overline{z}_k], \oplus_k[-\overline{G}_{k} + \overline{\mathcal{N}}(F_{k}) + \overline{I}_k], \oplus_k [-\overline{H}_{k} - \overline{\mathcal{N}}(F_{k}) + \overline{J}_k], \oplus_{k} \Tr(F_{k}), \\
    & \hspace{3cm} \Phi_{0}(\sigma) + M_{1}, -\Phi_{0}(\sigma) + M_{2}) \\
    X &= \widetilde{\text{diag}}(\sigma, \oplus_k F_{k}, \oplus_k G_{k}, \oplus_k H_k, \oplus_k z_k, \oplus_k I_k, \oplus_k J_k, \oplus_k \overline{G}_k, \oplus_k \overline{H}_k, \oplus_k \overline{z}_{k}, \oplus_k \overline{I}_k, \oplus_k \overline{J}_k, M_{1}, M_{2})
\end{align*}
It is straightforward to see the adjoint map of $\Psi$ in this case is
\begin{align*}
    \Psi^{\dagger}(Y) &= \text{diag}(\Phi_{0}^{\dagger}(W_{1} - W_{2}) + \sum_{k} \Phi_{\mathcal{P}_{k}}^{\dagger}(K_k - L_k), \oplus_k [\mathcal{N}_{k}^{\dagger}(L_k - K_k) + \overline{\mathcal{N}}^{\dagger}(\overline{K}_k - \overline{L}_k) + b_k \bbone_{\mathcal{W}}], \oplus_k [a_k \bbone_{\mathcal{W}} - K_k], \\
    & \hspace{0.2cm} \oplus_k [a_k \bbone_{\mathcal{W}} - L_k], \oplus_k a_k, \oplus_k K_k, \oplus_k L_k, \oplus_k [\overline{a}_{k} \bbone_{\mathcal{W}} - \overline{K}_k], \oplus_k [\overline{a}_k \bbone_{\mathcal{W}} - \overline{L}_k], \oplus_k \overline{a}_k, \oplus_k \overline{K}_k, \oplus_k \overline{L}_k, W_1 , W_2)
\end{align*}
where
\begin{align*}
    Y = \widetilde{\text{diag}}(\oplus_k a_k, \oplus_k K_k, \oplus_k L_k, \oplus_k \overline{a}_k, \oplus_k \overline{K}_k, \oplus_k \overline{L}_k, \oplus_k b_k, W_1, W_2)
\end{align*}

Finally, again because all of the $k$s are independent, this dual problem is ultimately simplified to:
\begin{equation}
        \begin{aligned}
                & {\text{maximize}} & & \sum_k \overline{f} \cdot \vec{\overline{z}}_k + (\epsilon' + \vec{\gamma}) \cdot \vec{y}_{1} + (\epsilon' - \vec{\gamma}) \cdot \vec{y}_{2} - \vec{\mu} \cdot \vec{a} - \vec{t} \cdot \vec{\overline{a}} - \|\vec{b}\|_{1} & \\
                & \text{subject to} & & \sum_i [y_{1}(i) - y_{2}(i)]\Gamma_{i} + \sum_k ( \sum_j z_{k}(j) \widetilde{\Gamma}^{k}_j) \preceq \nabla f_{\epsilon}(\rho) & \\
                &                   & & \overrightarrow{\overline{\mathcal{N}}^{\dagger}}(\vec{\overline{z}}_k) - \overrightarrow{\mathcal{N}^{\dagger}_{k}}(\vec{z}) \leq b_{k} \vec{1}_{\mathcal{W}} & \forall k \\
                &                   & & -a_k \vec{1}_{\mathcal{W}} \leq \vec{z}_{k} \leq  a_k \vec{1}_{\mathcal{W}} & \forall k\\
                &                   & & -\overline{a}_k \vec{1}_{\mathcal{W}} \leq \vec{\overline{z}}_{k} \leq  \overline{a}_k \vec{1}_{\mathcal{W}} & \forall k\\
                &                   & &  \vec{a},\vec{\overline{a}} \geq 0 \hspace{1cm} \vec{y}_{1},\vec{y}_{2} \leq 0
        \end{aligned}
\end{equation}
where $\vec{\mu'}$ is just the vector whose $k$-th entry is given by $\mu'_k$ and $z_{k},\bar{z}_{k}$ are not the variable in the definition of $X$ but is a simplification of the dual variable as defined in the same fashion as in Eqn. (\ref{eq:vecSDPFinal}). From these forms, it is clear that strong duality and tightness proofs follow from the single POVM case by indexing over the variable $k$ and scaling things properly.

\end{section}
 
\section{Derivations of Terms}\label{appendix:Terms}
In this section we derive the terms in the keyrate which differ from previous works. \\

Recall that an input $\xi$ is $\varepsilon_{\PE}$-securely filtered if the probability that Alice and Bob do not abort the parameter estimation subprotocol on input $\xi$ is less than $\varepsilon_{\PE}$. Given a bipartite measurement $\{\widetilde{\Gamma}_{j}\}$, by Born's rule, the measurement and a bipartite state $\sigma$ induce a probability distribution over measurement outcomes, $p$. Therefore, if one measures $\sigma$ $n$ times using $\{\widetilde{\Gamma}_{j}\}$ each time, it is sufficient to determine a distance between $p$ and the observed frequency distribution over measurement outcomes, $f$, such that the probability of obtaining a frequency distribution $||f-p||_{1} > \mu$ is less than $\varepsilon_{\PE}$. The following theorem captures this notion.
 
\begin{theorem}\label{thm:muProof} To construct a set of states, $\mathbf{S}_{\mu}$ (Eqn. \ref{eq:paramEstSet1}), such that the complement of the set, $\overline{\mathbf{S}_{\mu}}$, satisfies the property that $\forall \sigma \in \overline{\mathbf{S}_{\mu}}$, $\sigma^{\otimes m}$ is $\varepsilon_{\PE}$-securely filtered, it is sufficient that $\mu = \sqrt{2}\sqrt{\frac{\ln(1/\varepsilon_\PE) + |\Sigma|\ln(m + 1)}{m}}$. \end{theorem}
\begin{proof} 
By Theorem 11.2.1 of \cite{cover2006}, given an empirical probability distribution $f$ constructed from sampling i.i.d. random variables from a probability distribution $p$ which has $|\Sigma|$ outcomes,  
$$ \text{Pr}[D(f||p) > \epsilon] \leq 2^{-m(\epsilon - |\Sigma|\frac{\log_{2}(m+1)}{m})} $$
Furthermore, Lemma 11.6.1 of \cite{cover2006} states:
$$ \sqrt{2\ln 2 D(f||p)} \geq \|f - p \|_{1}$$
Therefore,
\begin{align*}
   & \text{Pr}\left [ \|f-p\|_{1} > \sqrt{2 \ln 2 \epsilon} \right] \\
   \leq & \text{Pr} \left[ \sqrt{2 \ln 2 D(f||p)} > \sqrt{2 \ln 2 \epsilon} \right] \\
   \leq & 2^{-m[\epsilon - |\Sigma|\frac{\log_{2}(m+1)}{m}]} \\
   \equiv & \varepsilon_{\PE}
\end{align*}
Then, except with probablity $\varepsilon_{\PE}$, $\|f-p\|_{1} \leq \sqrt{2 \ln 2 \epsilon} \equiv \mu$. We now just solve for $\mu$ using arithmetic:
\begin{align*}
     \varepsilon_{\PE} &= 2^{-m[\frac{\mu^{2}}{2\ln 2} - |\Sigma| \log_{2}(m+1)/m]}\\
    \Rightarrow \mu &= \sqrt{2}\sqrt{\frac{\ln(1/\varepsilon_{\PE}) + |\Sigma|\ln(m+1)}{m}}
\end{align*}
\end{proof}

\noindent \textbf{Derivation of $\delta(\bar{\varepsilon})$} \\

Our version of $\delta(\bar{\varepsilon})$ arises from the correction of a typo in Theorem 3.3.6 of \cite{renner05} and then stopping midway through the derivation of Corollary 3.3.7 of \cite{renner05} so as to have the prefactor $2\log_{2}(d+3)$ instead of $2d + 3$. As the typo in Theorem 3.3.6 was already noted in \cite{scarani2008}, we simply state the corrected theorem: \\

\noindent \textbf{Theorem 3.3.6 of \cite{renner05}:} Let $\rho \in D(\mathcal{H}_A \otimes \mathcal{H}_B)$, $\sigma_{B} \in D(\mathcal{H}_B)$, and $n \in \mathbb{N}$. Then for any $\varepsilon \geq 0$,
$$ \frac{1}{n} H^{\varepsilon}_{\min}(\rho^{\otimes n}|\sigma^{\otimes n}) \geq H(\rho) - H(\rho_{B}) - D(\rho_{B} || \sigma_{B}) - \delta$$
where $\delta = 2 \log_{2}(\text{rank}(\rho_{A}) + \Tr(\rho^{2}(\bbone_{A} \otimes \sigma_{B}^{-1}))+2)\sqrt{\frac{\log_{2}(2/\varepsilon)}{n}}$. \\

We now give our version of Corollary 3.3.7, which is the original proof but without adding looseness so as to write it in terms of max-entropy: \\

\begin{theorem} (Variation of Corollary 3.3.7 of \cite{renner05}) Let $\rho_{XB} \in D(\mathcal{H}_X \otimes \mathcal{H}_B)$ be a classical-quantum state. Then for any $\varepsilon \geq 0$, $$ \frac{1}{n} H^{\varepsilon}_{\min}(\rho_{XB}^{\otimes n} | \rho_{B}^{\otimes n}) \geq H(XB) - H(B) - \delta$$
where $\delta = 2\log_{2}(\text{rank}(\rho_{X}) + 3)\sqrt{\frac{\log_{2}(2/\varepsilon)}{n}}$. \end{theorem}

\begin{proof}
Without loss of generality, assume $\rho_{B}$ is invertible as the general statement follows by continuity. $$\bbone_{X} \otimes \rho_{B} - \rho_{XB} = \sum_{x \in X} (\bbone_{X} - \ket{x}\bra{x}) \otimes \rho_{B}^{x} \geq 0$$
Thus, by an operator inequality (Lemma B.5.4 of \cite{renner05}), we know
$$ \lambda_{\max}(\sqrt{\rho_{XB}}(\bbone_{X} \otimes \rho_{B}^{-1}) \sqrt{\rho_{XB}}) \leq 1 $$
As $\Tr(\rho_{XB}) = 1$,
$$ \Tr(\rho_{XB}^{2}(\bbone_{X} \otimes \rho_{B}^{-1})) = \Tr(\rho_{XB} (\sqrt{\rho_{XB}}(\bbone_{X} \otimes \rho_{B}^{-1}) \sqrt{\rho_{XB}})) \leq 1 $$
It therefore follows:
\begin{align*}
    \log_{2}(\text{rank}(\rho_{X}) + \Tr(\rho_{XB}^{2}(\bbone_{X} \otimes \rho_{B}^{-1})) + 2) \leq \log_{2}(\text{rank}(\rho_{X}) + 3)
\end{align*}
Plugging this value into Theorem 3.3.6 completes the proof.
\end{proof}
\begin{section}{Coherent Attack Analysis}\label{appendix:CoherentAttacks}
    
As noted in the main text, if one were to use the Finite Quantum de Finetti theorem to bound the coherent attack, there are a few minor changes from the presentation in the main text which is for collective attack. Namely, there is the introduction of another security term $\varepsilon_{\QdF}$, a different way to calculate the correction term $\delta(\bar{\varepsilon})$ as well as $\mu$, and two new parameters $r$ and $k$ which need to be chosen appropriately. To show that it can be handled, we briefly discuss where each change arises. \\

The first aspect is that the Quantum de Finetti theorem itself is a probabilistic statement about the distance between a subsystem of a large state and a state which is a convex combination of i.i.d. states. This probability, which we refer to as $\varepsilon_{\QdF}$, must then be included. In this case, we can therefore rewrite Theorem 6.5.1 of \cite{renner05} so $\varepsilon$ terms are explicit:\\

\noindent \textbf{Theorem 6.5.1 \cite{renner05}} Given a general QKD protocol as defined in the main text where a total of $N$ signals are transmitted, $m$ of the signals are used for parameter estimation, and $n$ of the signals are used for key generation, let $k \in \mathbb{N}$ and $bn + m + k = N$ where $b$ accounts for block-wise processing. Let $ \bar{\varepsilon}, \varepsilon_{\EC}, \varepsilon_{\PA}, \varepsilon_{\PE}, \varepsilon_{\QdF} > 0$. Then the QKD protocol is $(\varepsilon_{\QdF} + \varepsilon_{\PE} + \bar{\varepsilon} + \varepsilon_{\EC} + \varepsilon_{\PA})$-secure if the error correction is $\varepsilon_{\EC}$-secure and if 
\begin{align*}
    \ell \leq n[H_{\mu}(X|E) - \delta(\bar{\varepsilon})] - 2(m+k)\log_{2}(\text{dim}(\mathcal{H}_{A} \otimes \mathcal{H}_{B})) - \leak - 2 \log_{2}(\frac{2}{\varepsilon_{\PA}})
\end{align*}
where 
\begin{align} \mu &\equiv 2 \sqrt{h\left(\frac{r}{m}\right) + \frac{\log_{2}(1/\varepsilon_{\PE}) + |\Sigma| \log_{2}(\frac{m}{2}+1)}{m}} \label{eq:muNew} \\
\delta(\bar{\varepsilon}) &\equiv (\frac{5}{2}\log_{2}(d) + 4) \sqrt{h(r/n) + \frac{2}{n} \log_{2}(4/\bar{\varepsilon})} \label{eq:deltaNew} \\
r &\equiv \left(\frac{bn+m}{k}+1 \right) \left(2 \ln (\frac{2}{\varepsilon_{\QdF}}) + (\text{dim}(\mathcal{H}_{A} \otimes \mathcal{H}_{B})^{2} \ln(k) \right) - 1 \leq N \label{eq:r}
\end{align}
where $d$ is the size of the alphabet for Alice and Bob's output key.
\begin{proof}
See \cite{renner05}.
\end{proof}
As can be seen from the statement of the theorem, the key rate will be lower from that of the collective attack as the correction term $\delta(\bar{\varepsilon})$ and variation bound $\mu$ will be larger for any fixed $m$ for the finite key analysis. This is largely due to the binary entropy terms which depend on $r$. To make $r$ small, one must either let $\varepsilon_{\QdF}$ be large, or sacrifice many of the $N$ signals to make $k$ large. Physically, this `sacrifice' is to throw out a large portion of the signal states to make the rest of the system close enough to a mixture of i.i.d. signals.

Given this theorem, all one needs to do to use our numerical solver with the Finite Quantum de Finetti theorem is replace the variation bound in Eqn. \ref{eq:mu} with Eqn. \ref{eq:muNew}, the correction term from Eqn. \ref{eq:delta} with Eqn. \ref{eq:deltaNew}, add the $-2(m+k)\log_{2}(\text{dim}(\mathcal{H}_{A} \otimes \mathcal{H}_{B}))$ term to calculating the key length, and optimize over $k$ such that $r \leq N$. 

Finally, we note in the case one is interested in proving security for a prepare-and-measure QKD protocol and therefore needs to apply source-replacement scheme, one must introduce an extra $\varepsilon$-term as explained in Remark 4.3.3. of \cite{renner05}.
\end{section} 

\begin{section}{Post-processing Maps for Examples}\label{app:postProc}
In this section we provide the post-processing maps, $\mathcal{G}$ for each example for completeness. As explained in \cite{winick2018}, the post-processing map $\mathcal{G}$ can be decomposed into three operations on a state $\rho$ (see Appendix A of \cite{lin2019} for an in-depth derivation):
\begin{itemize}
    \item The isometric quantum channel $\mathcal{A}$ which represents the measurements of Alice and Bob as well as their partitioning of resulting data into public and private information.
    \item A projection $\Pi$ on their public data which represents the general sifting step.
    \item A partial isometry $V$ which acts on the subspace spanned by $\Pi$ which models the key map applied by Alice or Bob.
\end{itemize}
Then from these $\mathcal{G}$ is defined as $\mathcal{G}(\cdot) = V \Pi \mathcal{A}(\cdot) \Pi V^{\dagger}$. Lastly we note that the channel $\mathcal{A}(\cdot) = \sum_{a,b} (K^{A}_{a} \otimes K^{B}_{b})(\cdot)(K^{A}_{a} \otimes K^{B}_{b})^{\dagger}$ where $K^{A}_{a} = \sum_{{\alpha}_{a}} \sqrt{P^{A}_{a,\alpha_{a}}} \otimes \ket{a}_{\widetilde{A}} \otimes \ket{\alpha_{a}}_{\overline{A}}$, $K^{B}_{b} = \sum_{{\beta}_{b}} \sqrt{P^{B}_{b,\beta_{b}}} \otimes \ket{b}_{\widetilde{B}} \otimes \ket{\beta_{b}}_{\overline{B}}$, where the spaces with a tilde denote public information and spaces with a bar denote private information as in Fig. \ref{fig1:GeneralProtocol}. $a_{\alpha}$ denotes the outcome $\alpha$ given public announcement $a$, and $P^{A}_{a, \alpha_{a}}$ denotes the (fine-grained) measurement Alice would have done to have public information $a$ and private information $\alpha$. The notation is the same for Bob. We refer to Appendix A of \cite{lin2019} for a further discussion of the post-processing framework. 

\subsection{Single-Photon BB84}
The examples in Sections \ref{subsec:PhaseBB84} and \ref{subsec:rotatedBB84} use the same map $\mathcal{G}$. In singe-photon BB84, Alice and Bob perform von Neumann measurements in the $Z$ and $X$ bases with probabilities $p_{z}$ and $1-p_z$ respectively, the public information Alice and Bob announce are what bases they measure in,  the private information is what outcome they got (represented by a $0$ or $1$) in both bases, and the sifting throws out any measurement where Alice and Bob did not use the same basis. Lastly we note that in Sections \ref{subsec:PhaseBB84} and \ref{subsec:rotatedBB84}, we assumed Alice only performs the key map in the $Z$-basis. Therefore we have the following definitions for constructing the $\mathcal{G}$ map:
\begin{align*}
    K^A_{Z} &= \sqrt{p_z} \ket{0}\bra{0}_{A} \otimes \ket{0}_{\widetilde{A}} \otimes \ket{0}_{\overline{A}} + \sqrt{p_z} \ket{1}\bra{1}_{A} \otimes \ket{0}_{\widetilde{A}} \otimes \ket{1}_{\overline{A}} \\
    K^A_{X} &= \sqrt{1-p_z} \ket{+}\bra{+}_{A} \otimes \ket{1}_{\widetilde{A}} \otimes \ket{0}_{\overline{A}} + \sqrt{1-p_z} \ket{-}\bra{-}_{A} \otimes \ket{1}_{\widetilde{A}} \otimes \ket{1}_{\overline{A}} \\
    K^B_{Z} &= \sqrt{p_z}\ket{0}\bra{0}_{B} \otimes \ket{0}_{\widetilde{B}} \otimes \ket{0}_{\overline{B}} + \sqrt{p_z}\ket{1}\bra{1}_{B} \otimes \ket{0}_{\widetilde{B}} \otimes \ket{1}_{\overline{B}} \\
    K^B_{X} &= \sqrt{1-p_z}\ket{+}\bra{+}_{B} \otimes \ket{1}_{\widetilde{B}} \otimes \ket{0}_{\overline{B}} + \sqrt{1-p_z}\ket{-}\bra{-}_{B} \otimes \ket{0}_{\widetilde{B}} \otimes \ket{1}_{\overline{B}} \\
    \Pi &= \ket{0}\bra{0}_{\widetilde{A}} \otimes \ket{0}\bra{0}_{\widetilde{B}} + \ket{1}\bra{1}_{\widetilde{A}} \otimes \ket{1}\bra{1}_{\widetilde{B}} \\
    V &= \ket{0}_{R} \otimes \ket{0}\bra{0}_{\widetilde{A}} \otimes \ket{0}\bra{0}_{\overline{A}} \otimes \ket{0}\bra{0}_{\widetilde{B}}  + \ket{1}_{R} \otimes \ket{0}\bra{0}_{\widetilde{A}} \otimes \ket{1}\bra{1}_{\overline{A}} \otimes \ket{0}\bra{0}_{\widetilde{B}} 
\end{align*}
We note that while we used the source-replacement scheme, we used the Gram-Schmidt process to return Alice's space to the original size as explained in \cite{ferenczi12a}, which in this case reconstructs Alice's original POVM.

\subsection{MDI BB84}
For MDI BB84, as we consider the case where Alice and Bob only distill key from the Z basis, using the source-replacement scheme on both Alice's and Bob's sources and the simplification rules explained in Appendix A of \cite{lin2019}, there is only one Kraus operator for the entire map $\mathcal{G}$:
$$ K_{Z} = (\ket{0}_{R} \otimes \ket{0}\bra{0}_{A} + \ket{1}_{R} \otimes \ket{1}\bra{1}_{A}) \otimes (\ket{0}\bra{0}_{B} + \ket{1}\bra{1}_{B}) \otimes (\ket{0}\bra{0}_{C} + \ket{1}\bra{1}_{C}) $$

\subsection{Discrete-phase-randomized BB84}
In the discrete-phase-randomized BB84, we begin from the use of the squashing model which results in Alice preparing 4 states for each global phase, and Bob having the 5-outcome POVM described in Section \ref{subsec:DPRBB84}. Then by the source-replacement scheme on Alice, Alice's portion of the signal is a $4c$-dimensional Hilbert space $\mathcal{H}_{A}$ where $c$ is the number of global phases Alice uses. In other words, $\mathcal{H}_{A} \cong \oplus_{c} \mathcal{H}_{4}$ where $\mathcal{H}_{4}$ is a 4-dimensional Hilbert space and $\oplus$ is the direct sum. To make the expression of the Kraus operators concise, define the projector $\Pi_{n} = \dyad{n}{n}$ where $n \in \{0,1,2,3\}$. Then, using that Alice performs the key map along with the simplifications from Appendix A of \cite{lin2019}, we have two Kraus operators which describe the action of $\mathcal{G}$:
\begin{align*}
    K_{Z} &= \ket{0}_{R} \otimes \left (\bigoplus_{c}(\Pi_{0}) \right) \otimes \sqrt{p_{z}}(\ket{0}\bra{0}_{B} + \ket{1}\bra{1}_{B}) \otimes \ket{0}_{\widetilde{A}} + \ket{1}_{R} \otimes \left (\bigoplus_{c}(\Pi_{1}) \right) \otimes \sqrt{p_{z}}(\ket{0}\bra{0}_{B} + \ket{1}\bra{1}_{B}) \otimes \ket{0}_{\widetilde{A}} \\
    K_{X} &= \ket{0}_{R} \otimes \left (\bigoplus_{c}(\Pi_{2}) \right) \otimes \sqrt{1-p_{z}}(\ket{+}\bra{+}_{B} + \ket{-}\bra{-}_{B}) \otimes \ket{1}_{\widetilde{A}} \\
    & \hspace{2cm} + \ket{1}_{R} \otimes \left (\bigoplus_{c}(\Pi_{3}) \right) \otimes \sqrt{1-p_{z}}(\ket{+}\bra{+}_{B} + \ket{-}\bra{-}_{B}) \otimes \ket{1}_{\widetilde{A}} \\
\end{align*}
where $\oplus_{c} \Pi_{n}$ is well defined for all $n$ as $\Pi_{n} \in \mathcal{H}_{4}$ and $\mathcal{H}_{A} \cong \oplus_{c} \mathcal{H}_{4}$. 
\end{section}

\bibliographystyle{apsrev4-2}
\bibliography{main}
    
\end{document}